\documentclass[letterpaper,11pt]{article}

\newcommand{\tit}{A Practical Oblivious Map Data Structure\texorpdfstring{\\}{}
with Secure Deletion and History Independence}

\usepackage[margin=1 in,headsep=15 pt]{geometry}

\usepackage[iso,english]{isodate}
\usepackage{calc}
\usepackage{xspace}
\usepackage{url}
\usepackage{cite}
\usepackage{graphicx}
\usepackage{fancyvrb}
\usepackage[english]{babel}
\usepackage[usenames,dvipsnames]{color}
\usepackage{breakurl}
\usepackage{framed}
\usepackage{booktabs}
\usepackage{courier}   %
\usepackage{times}
\usepackage{amsmath,amsthm,amssymb}
\usepackage{amsfonts}
\usepackage{enumitem}
\usepackage[font=small,labelfont=bf]{caption}
\usepackage{thmtools,thm-restate}
\usepackage{fancyhdr}
\usepackage{tikz}
\usepackage{appendix}
\usepackage{authblk}
\usepackage[numbers]{natbib}
\usepackage{hyperref}
\usepackage[capitalise,noabbrev]{cleveref}
\usepackage{algorithmicx}
\usepackage{algpseudocode}

\graphicspath{{figures/}{.}}

\newif\ifpreprint
\preprinttrue

\makeatletter
\def\maxwidth#1{\ifdim\Gin@nat@width>#1 #1\else\Gin@nat@width\fi}
\makeatother

\newcommand{\BT}{\begin{theorem}}
\newcommand{\ET}{\end{theorem}}

\newcommand{\BD}{\begin{definition}}
\newcommand{\ED}{\end{definition}}

\newcommand{\BCR}{\begin{corollary}}
\newcommand{\ECR}{\end{corollary}}

\newcommand{\BEX}{\begin{example}}
\newcommand{\EEX}{\end{example}}

\newcommand{\BL}{\begin{lemma}}
\newcommand{\EL}{\end{lemma}}

\newcommand{\BP}{\begin{proposition}}
\newcommand{\EP}{\end{proposition}}

\newcommand{\BCM}{\begin{claim}}
\newcommand{\ECM}{\end{claim}}

\newcommand{\BPF}{\begin{proof}}
\newcommand{\EPF}{\end{proof}}
\newcommand{\BEN}{\begin{enumerate}}
\newcommand{\EEN}{\end{enumerate}}

\newcommand{\BI}{\begin{itemize}}
\newcommand{\EI}{\end{itemize}}

\newcommand{\BO}{\begin{observation}}
\newcommand{\EO}{\end{observation}}

\newcommand{\BDS}{\begin{description}}
\newcommand{\EDS}{\end{description}}

\def\EE{{\mathbb E}}       %
\def\AAA{{\mathcal A}}     %
\def\DDD{\ensuremath{\mathcal{D}}}     %
\def\zo{{\{0,1\}}}    %

\def\from{{\leftarrow\,}}
\def\E{\mathop{\mathbb{E}}\displaylimits}   %

\newcommand{\me}{\ensuremath{\mathrm{e}}}

\newcommand{\Hash}{{\sf Hash}}

\newcommand{\secparam}{\ensuremath{\lambda}}
\newcommand{\hashparam}{\ensuremath{\gamma}}

\newcommand{\nodesize}{\ensuremath{\mathsf{nodesize}}}

\newcommand{\gconst}{\ensuremath{c_0}}

\newcommand{\ignore}[1]{}

\def\Enc{\mathsf{Enc}}

\newcommand{\negl}{{\sf negl}}

\newcounter{defcounter}
\newlength{\protowidth}

\newcommand{\advan}{\ensuremath{\mathbf{Adv}}}

\newcommand{\olrk}[1]{%
   \ifx\nursymbol#1\else\!\!\mskip4.5mu plus 0.5mu\left(#1\right)\fi}
\newcommand{\elrk}[1]{%
   \ifx\nursymbol#1\else%
        \!\!\mskip4.5mu plus0.5mu\left[\mskip2.5mu plus0.5mu #1\right]\fi}

\usepackage{framed}

\newcommand{\st}{\textsc{st}}

\def\init{\ensuremath{\mathsf{init}}}

\def\delete{\ensuremath{\mathsf{delete}}}

\newcommand{\set}{\ensuremath{\mathsf{set}}}
\newcommand{\get}{\ensuremath{\mathsf{get}}}
\newcommand{\hirbinit}{\ensuremath{\mathsf{hirbinit}}}
\newcommand{\chooseheight}{\ensuremath{\mathsf{chooseheight}}}
\newcommand{\hirbpath}{\ensuremath{\mathsf{HIRBpath}}}
\newcommand{\dskey}{\ensuremath{\mathsf{label}}}
\newcommand{\dsval}{\ensuremath{\mathsf{value}}}
\newcommand{\child}{\ensuremath{\mathsf{child}}}

\def\mem{\ensuremath{\mathsf{em}}}
\def\hdd{\ensuremath{\mathsf{ps}}}
\def\op{\ensuremath{\mathsf{op}}}
\def\acc{\ensuremath{\mathsf{acc}}}
\def\data{\ensuremath{\mathsf{data}}}
\def\sdel{\ensuremath{{\mathsf{sdel}}}}
\def\obh{\ensuremath{{\mathsf{obl\mbox{-}hi}}}}
\def\ob{\ensuremath{{\mathsf{obl}}}}
\def\init{\ensuremath{\mathsf{Init}}}

\def\hi{\ensuremath{{\mathsf{hi}}}\xspace}

\newcommand{\tupsep}{\ensuremath{,\ }}

\def\exphi{\mathsf{EXP}_{\AAA}^{\obh}(\DDD, \secparam, n, 1, b)}
\def\exphiz{\mathsf{EXP}_{\AAA}^{\obh}(\DDD, \secparam, n, 1, 0)}
\def\exphio{\mathsf{EXP}_{\AAA}^{\obh}(\DDD, \secparam, n, 1, 1)}

\def\advhi{\advan_\AAA^\hi(\DDD,\secparam, n)}

\def\expsdel{\mathsf{EXP}_{\AAA_1, \AAA_2, \AAA_3}^{\sdel}(\DDD, \secparam, n, b)}
\def\expsdelz{\mathsf{EXP}_{\AAA}^{\sdel}(\DDD, \secparam, n, 0)}
\def\expsdelo{\mathsf{EXP}_{\AAA}^{\sdel}(\DDD, \secparam, n, 1)}
\def\advsdel{\advan_\AAA^\sdel(\DDD,\secparam, n)}

\def\expob{\mathsf{EXP}_{\AAA}^{\obh}(\DDD, \secparam, n, 0, b)}
\def\expobz{\mathsf{EXP}_{\AAA}^{\obh}(\DDD, \secparam, n, 0, 0)}
\def\expobo{\mathsf{EXP}_{\AAA}^{\obh}(\DDD, \secparam, n, 0, 1)}
\def\advob{\advan_\AAA^\ob(\DDD,\secparam, n)}

\def\expobh{\mathsf{EXP}_{\AAA_1, \AAA_2}^{\obh}(\DDD, \secparam, n, h, b)}

\renewcommand{\paragraph}[1]{\medskip \noindent \textbf{#1}~}
\renewcommand{\Vec}[1]{\overrightarrow{{#1}}}

\def\evict{\ensuremath{{\sf evict}}}
\def\writeback{\ensuremath{{\sf writeback}}}
\def\idgen{\ensuremath{{\sf idgen}}}
\def\stash{\ensuremath{{stash}}}
\def\path{\ensuremath{{\sf loc}}}

\def\nil{\ensuremath{{\sf nil}}}

\newcommand{\visoram}[0]{%
\begin{figure}[t]
\centering
\includegraphics[width=\minof{\linewidth}{4 in}]{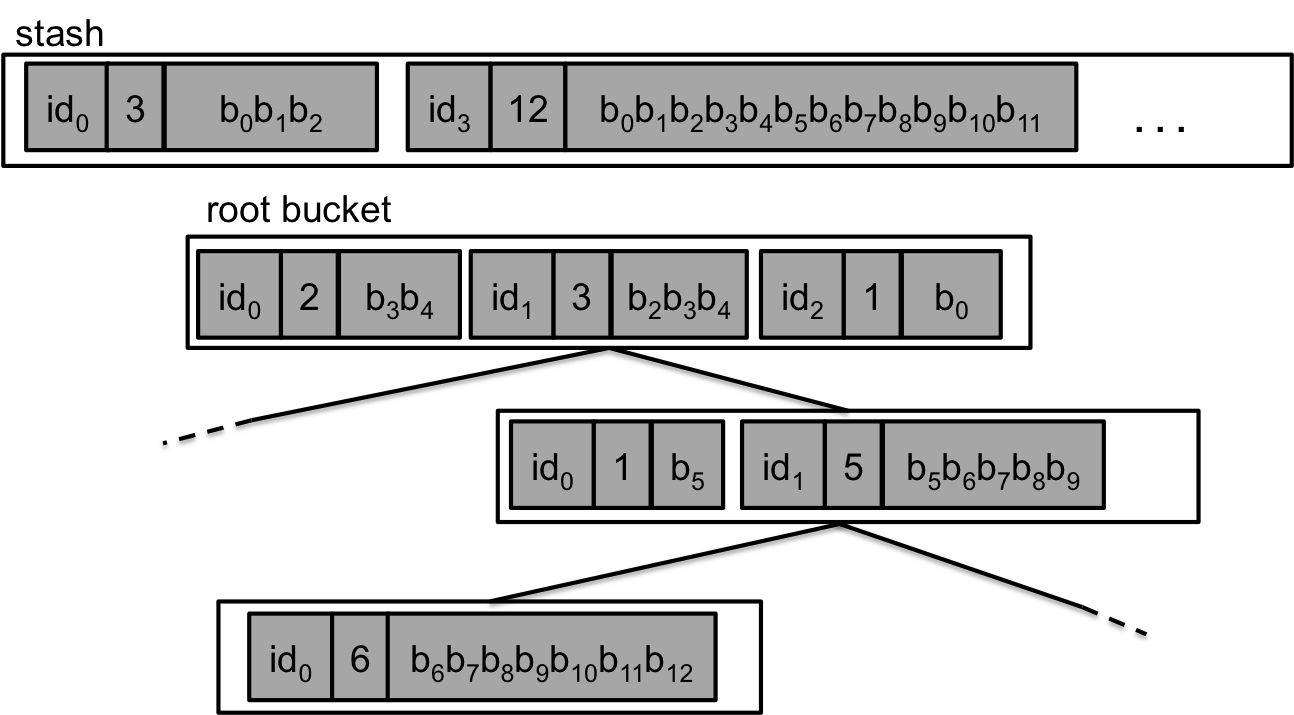}
\caption{A sample vORAM state with partial blocks with $id_0$, $id_1$,
  $id_2$, $id_3$: Note that the partial blocks for $id_0$ are
  opportunistically filled up the vORAM from leaf to root and then
  remaining partial blocks are placed in the stash.}
\label{fig:voram-vis}
\end{figure}%
}

\newcommand{\graphloadutil}[0]{%
\begin{figure}[t]
\centering
\includegraphics[width=\minof{0.88\linewidth}{5 in}]{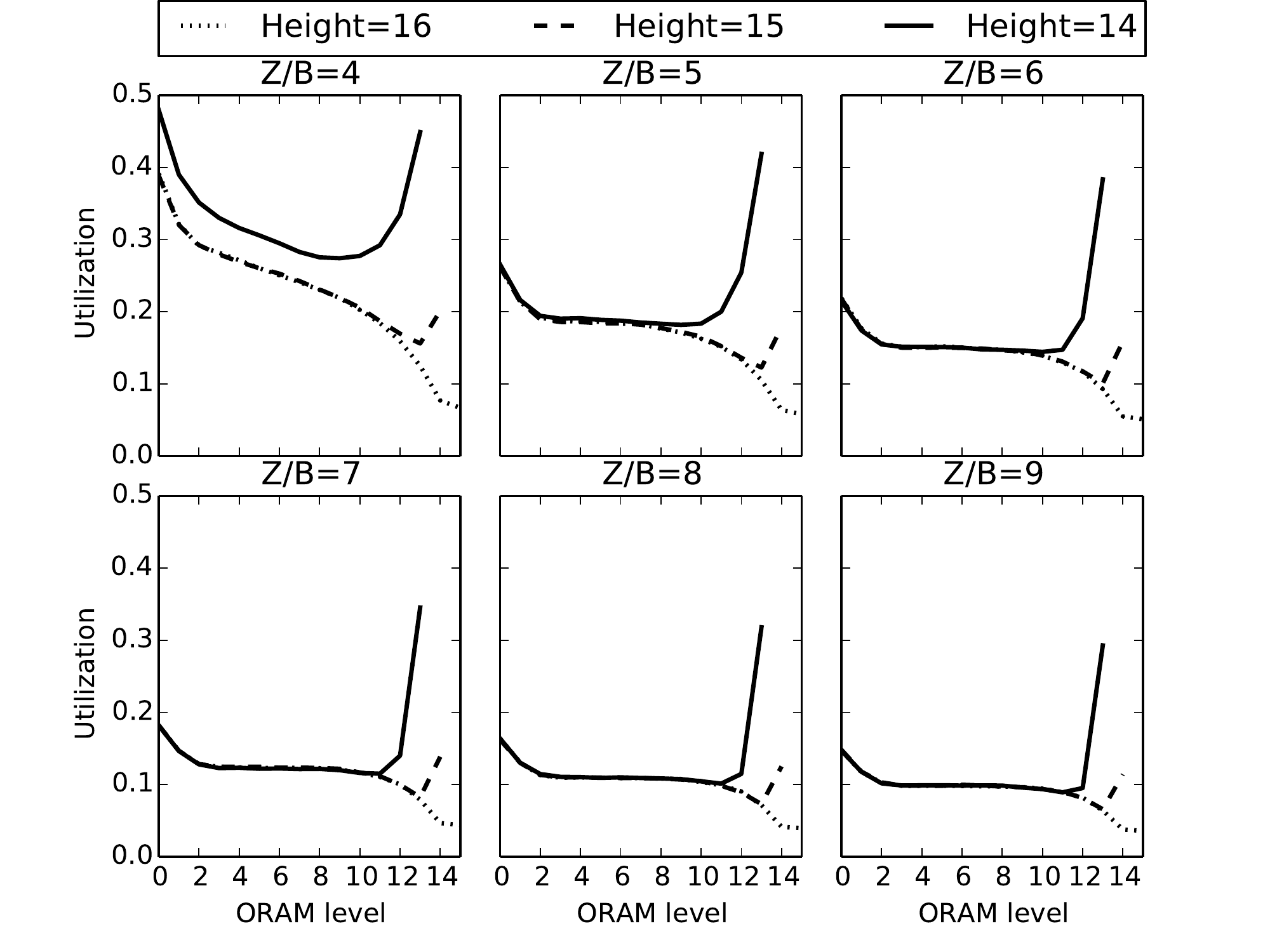}
\caption{Utilization at different levels of ORAM}
\label{fig:loadutil}
\end{figure}%
}

\newcommand{\graphmaxstash}[0]{%
\begin{figure}
\centering
\includegraphics[width=\minof{4.5in}{\linewidth}]{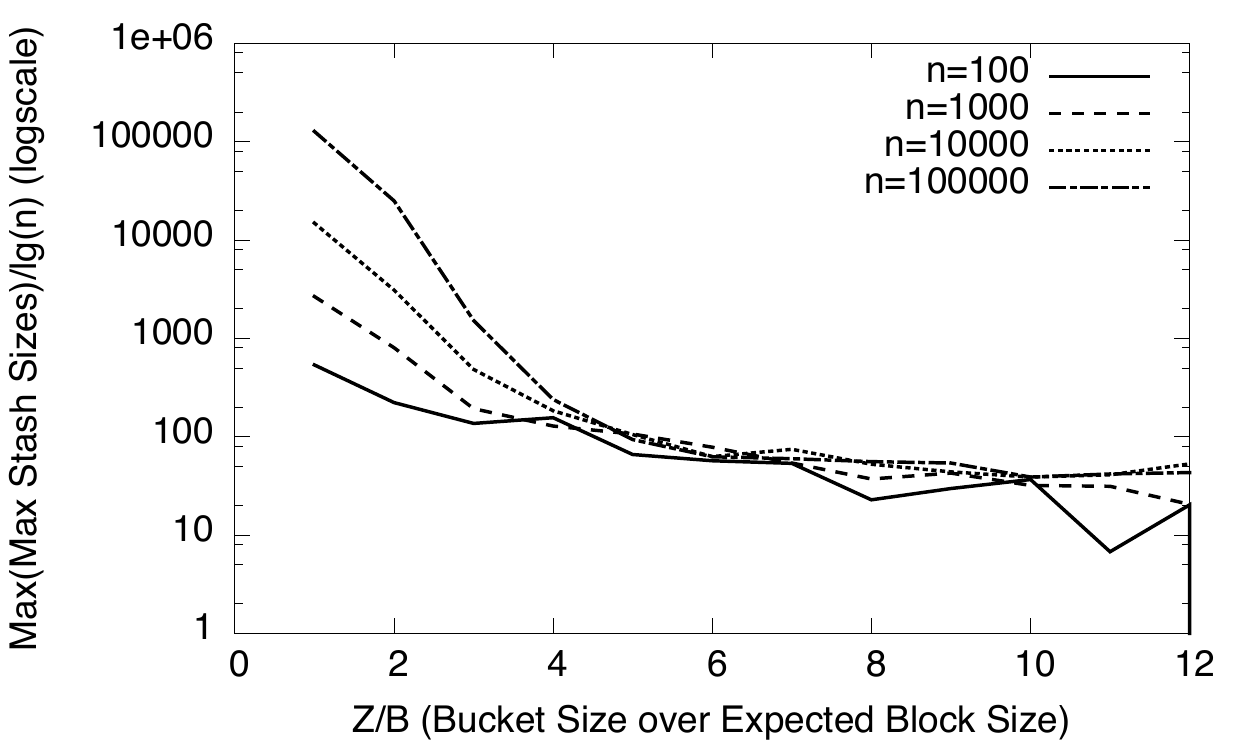}
\caption{Maximum stash size, scaled by $\log n$, observed across 50
  simulations of a vORAM for various $Z/B$ values.}
\label{fig:maxstash}
\end{figure}%
}

\newcommand{\graphmttf}[0]{
\begin{figure}
\centering
\includegraphics[width=\minof{4.5in}{\linewidth}]{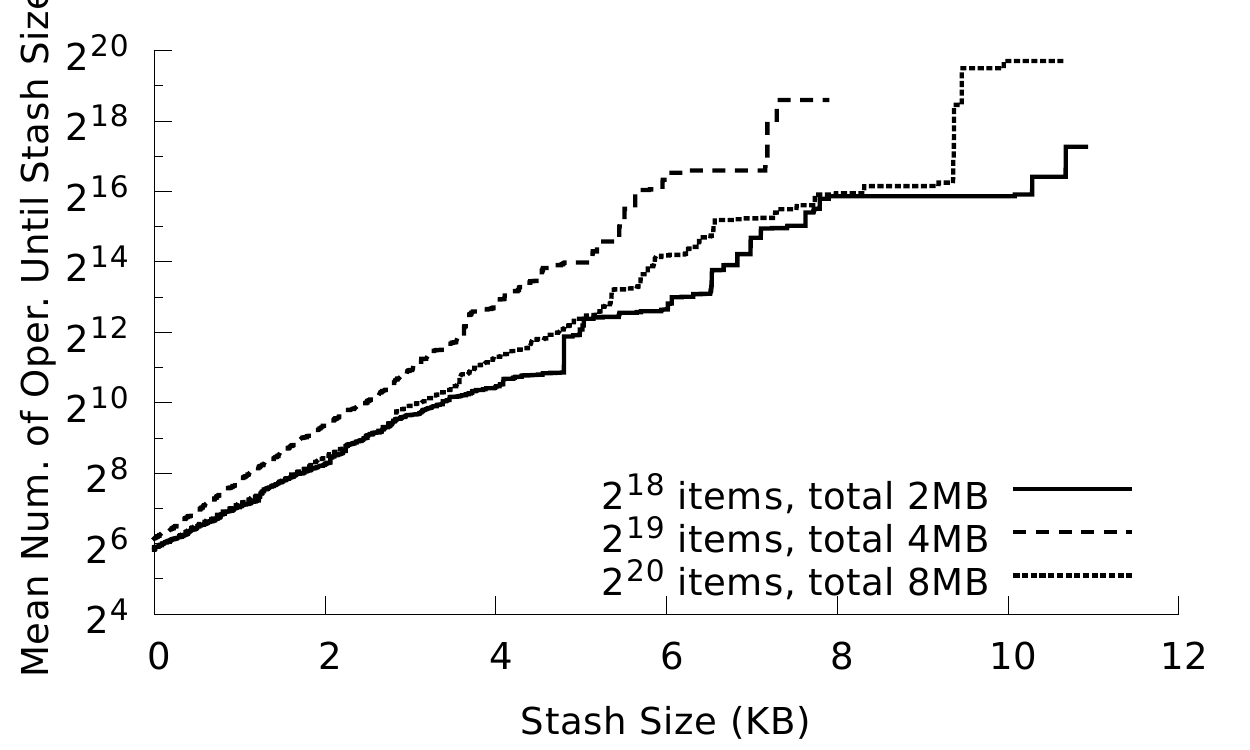}
\caption{Average time until stash overflow, for varying vORAM and stash
sizes. Stash size is linear-scale, number of operations in log-scale.
Higher is better.
For each vORAM size $n$, we performed $2n$ operations to gather
sufficient experimental data.}
\label{fig:mttf}|
\end{figure}%
}

\newcommand{\graphaccesstime}[0]{%
\begin{figure}[t]
\centering
\includegraphics[width=\minof{4.5in}{\linewidth}]{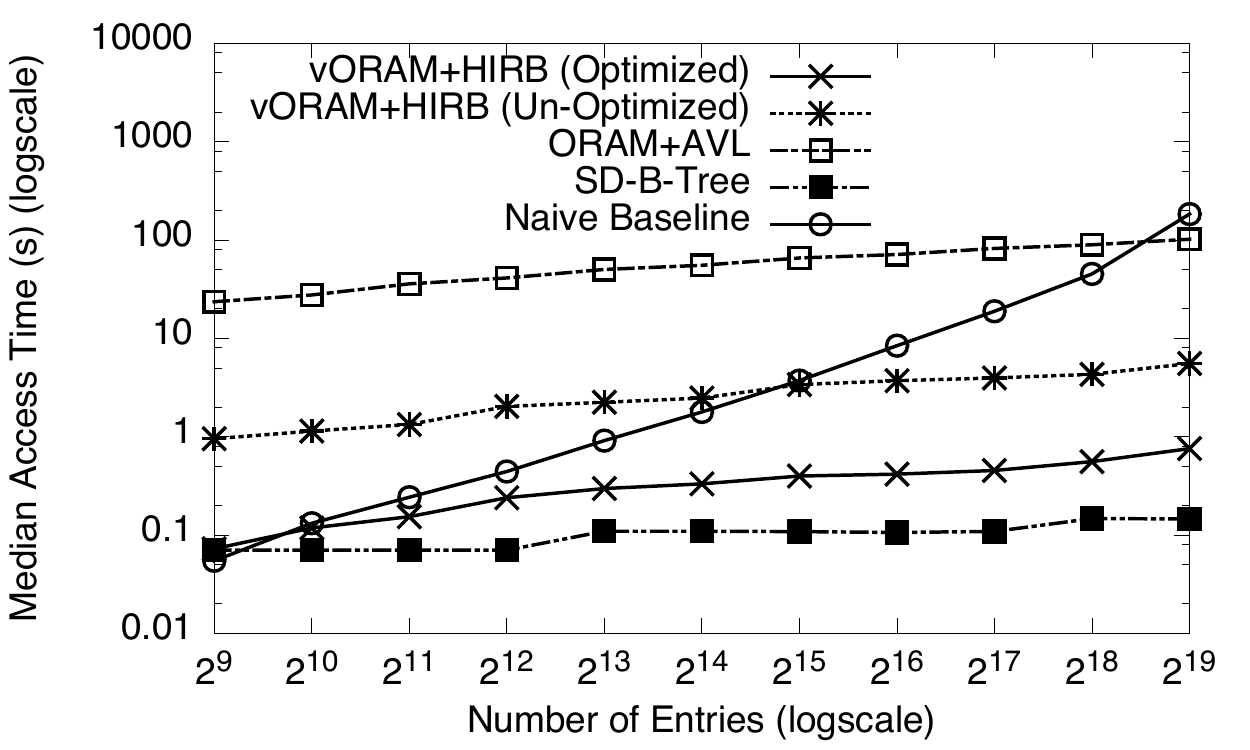}
\caption{Median of 100 access times for different number of entries}
\label{fig:accesstime}
\end{figure}%
}

\newcommand{\tablecomm}[0]{%
\begin{table}[tbp]
\centering

\begin{tabular}{|l|rrr|}
\hline \multicolumn{4}{|c|}{} \\[-3 mm]
\multicolumn{4}{|c|}{Size $2^{10}$} \\
\hline
\multicolumn{1}{|c|}{}
  & \multicolumn{1}{c}{Total storage}
  & \multicolumn{1}{c}{Bandwidth}
  & \multicolumn{1}{c|}{Rounds}
  \\
\hline
Naive baseline& 8.2 KB& 8.2 KB& 1\\
Secure deletion B-tree& 36.9 KB& 12.3 KB& 2\\
ORAM+AVL& 8.4 MB& 4.0 MB& 968\\
\textbf{vORAM+HIRB}& \textbf{127.0 KB}& \textbf{102.4 KB}& \textbf{3}\\
\hline
\hline \multicolumn{4}{|c|}{} \\[-3 mm]
\multicolumn{4}{|c|}{Size $2^{15}$} \\
\hline
\multicolumn{1}{|c|}{}
  & \multicolumn{1}{c}{Total storage}
  & \multicolumn{1}{c}{Bandwidth}
  & \multicolumn{1}{c|}{Rounds}
  \\
\hline
Naive baseline& 262.1 KB& 262.1 KB& 1\\
Secure deletion B-tree& 1.1 MB& 20.5 KB& 3\\
ORAM+AVL& 268.4 MB& 8.6 MB& 2096\\
\textbf{vORAM+HIRB}& \textbf{4.2 MB}& \textbf{286.7 KB}& \textbf{4}\\
\hline
\hline \multicolumn{4}{|c|}{} \\[-3 mm]
\multicolumn{4}{|c|}{Size $2^{20}$} \\
\hline
\multicolumn{1}{|c|}{}
  & \multicolumn{1}{c}{Total storage}
  & \multicolumn{1}{c}{Bandwidth}
  & \multicolumn{1}{c|}{Rounds}
  \\
\hline
Naive baseline& 8.4 MB& 8.4 MB& 1\\
Secure deletion B-tree& 33.8 MB& 20.5 KB& 3\\
ORAM+AVL& 8.6 GB& 15.1 MB& 3675\\
\textbf{vORAM+HIRB}& \textbf{134.2 MB}& \textbf{553.0 KB}& \textbf{5}\\
\hline
\hline \multicolumn{4}{|c|}{} \\[-3 mm]
\multicolumn{4}{|c|}{Size $2^{25}$} \\
\hline
\multicolumn{1}{|c|}{}
  & \multicolumn{1}{c}{Total storage}
  & \multicolumn{1}{c}{Bandwidth}
  & \multicolumn{1}{c|}{Rounds}
  \\
\hline
Naive baseline& 268.4 MB& 268.4 MB& 1\\
Secure deletion B-tree& 1.1 GB& 28.7 KB& 4\\
ORAM+AVL& 274.9 GB& 23.2 MB& 5668\\
\textbf{vORAM+HIRB}& \textbf{4.3 GB}& \textbf{901.1 KB}& \textbf{6}\\
\hline
\hline \multicolumn{4}{|c|}{} \\[-3 mm]
\multicolumn{4}{|c|}{Size $2^{30}$} \\
\hline
\multicolumn{1}{|c|}{}
  & \multicolumn{1}{c}{Total storage}
  & \multicolumn{1}{c}{Bandwidth}
  & \multicolumn{1}{c|}{Rounds}
  \\
\hline
Naive baseline& 8.6 GB& 8.6 GB& 1\\
Secure deletion B-tree& 34.6 GB& 36.9 KB& 5\\
ORAM+AVL& 8.8 TB& 33.3 MB& 8122\\
\textbf{vORAM+HIRB}& \textbf{137.4 GB}& \textbf{1.5 MB}& \textbf{8}\\
\hline
\end{tabular}

\smallskip
\caption{Storage and communication cost comparisons.
Total storage is the amount of space required for the server,
and the bandwidth and rounds are counted \emph{per operation}.
Each stored item consists of a 4-byte label and 4-byte value.%
\label{tab:comm}}
\end{table}%
}

\newcommand{\comment}[1]{}
\newtheorem{definition}{Definition}
\newtheorem{theorem}{Theorem}
\newtheorem{lemma}[theorem]{Lemma}

\newcommand{\ins}{\ensuremath{\mathsf{insert}}}
\newcommand{\remove}{\ensuremath{\mathsf{remove}}}
\newcommand{\upd}{\ensuremath{\mathsf{update}}}

\usetikzlibrary{calc,positioning,shapes,decorations.pathreplacing}

\tikzset{
short/.style={draw,rectangle,text height=3pt,text depth=13pt,
  text width=7pt,align=center,fill=gray!30},
long/.style={short,text width=1.5cm},
medium/.style={short,text width=7mm}
}

\hypersetup{
  bookmarks=false, %
  colorlinks=true,%
  citecolor=black,%
  filecolor=black,%
  linkcolor=black,%
  urlcolor=black,%
  pagebackref=true,%
  pdftitle={\tit},
  pdfauthor={Daniel S. Roche, Adam J. Aviv, Seung Geol Choi},
}

\newcommand{\fullversion}[1]{#1}

\begin{document}

\title{\tit}

\author{Daniel S.\ Roche}
\author{Adam J.\ Aviv}
\author{Seung Geol Choi}
\affil{United States Naval Academy\\
{\tt [roche,aviv,choi]@usna.edu}}

\date{{\origdate\today}}

\maketitle

\begin{abstract}
  We present a new oblivious RAM that supports variable-sized storage
  blocks (vORAM), which is the first ORAM to allow varying block sizes
  without trivial padding.
  We also present a new history-independent data structure (a HIRB
  tree) that can be stored within a vORAM.
  Together, this construction provides an efficient and practical
  oblivious data structure (ODS) for a key/value map, and goes further
  to provide an additional privacy guarantee as
  compared to prior ODS maps: even upon client compromise,
  deleted data and the history of old operations remain hidden to the
  attacker. 
  We implement and measure the performance of our system using Amazon
  Web Services, and the single-operation time for a realistic
  database (up to $2^{18}$ entries) is less than 1 second. This
  represents a 100x speed-up compared to
  the current best oblivious map data structure
  (which provides neither secure deletion nor history independence)
  by Wang et al.~(CCS 14).
  
\end{abstract}

\section{Introduction}
\label{sec:intro}

\subsection{Motivation}

Increasingly, organizations and individuals are storing large amounts
of data in remote, shared cloud servers. For sensitive data, it is
important to protect the privacy not only of the data itself but also of the
access to the metadata that may contain \emph{which} records have been
accessed and \emph{when},
thereby revealing properties of the underlying data, even if that data
is encrypted.  There are multiple points of potential information
leakage in this setting: an adversary could observe network communication
between the client and server; an adversary could compromise the cloud itself,
observing the data stored at the server, possibly including mirrored copies or
backups; an adversary could observe the
computations performed by the remote server; the adversary may compromise the
locally-stored client data; or, finally, the adversary may compromise the data
in multiple ways, e.g., a complete compromise of both the remotely stored cloud
storage and locally-stored client storage\footnote{We assume an
honest-but-curious server throughout, and leave achieving an ODS with
malicious servers as an open problem.}.

While a complete compromise will inevitably reveal private data, we
seek data storage mechanisms which maximize privacy while maintaining
reasonable, practical efficiency, at any level of compromise. For
generality, we assume a computationally-limited server which may only
store and retrieve blocks of raw data, and we focus on the most basic
(and perhaps most important) data structure: a key/value map.

\paragraph{Oblivious RAM (ORAM).}
With a computationally-limited server, the \emph{access pattern} of
client-server communication reveals the entire history of the remote
data store.
This access pattern, even if the actual data is encrypted, may
leak sensitive information about the underlying stored data, such as
keyword search queries or encryption keys~\cite{ASPLOS:ZhuZhaPan04,
  NDSS:IslKuzKan12, dautrich2013compromising}.

A {\em generic} solution to protect against access pattern leakage is oblivious
RAM (ORAM)~\cite{GO96}, which obscures the operation being
performed (read/write), the address on which it operates, and the
contents of the underlying data.
Any program (with
the possible necessity of some padding) can be executed using an ORAM to hide
the access patterns to the underlying data. 

A great number of ORAM schemes have been recently proposed, most aiming
to improve the efficiency as it relates to the recursive index
structure, which is typically
required to store the hidden locations of items within the ORAM
(for example
\cite{TCC:DamMelNie11,SODA:GMOT12,SODA:KusLuOst12,NDSS:SSS12,CCS:SDSFRY13,usenix:DSS14}
and references therein). 
However, an important aspect overlooked by previous work is the
\emph{size} of data items themselves. The vORAM construction we propose
provides an affirmative answer to the following question:

\begin{quotation}
  \noindent\em
  Can an oblivious RAM hide the size of varying-sized items, with
  greater efficiency than that achieved by trivial padding?
\end{quotation}

\paragraph{Oblivious data structure (ODS).} 
Recently, Wang et al.~\cite{CCS:WNLCSS14} showed that it is possible to provide
obliviousness more efficiently if the specifics of the target program are
considered.
In particular, among other results, Wang et al.~achieved an oblivious data
structure (ODS) scheme for a key-value map, by constructing an AVL tree {\em on
a non-recursive ORAM without using the position map}. Their scheme requires
$\tilde O(\log n)$ ORAM blocks of client storage, where $n$ is the maximum
number of allowable data items. More importantly, due to lack of position map
lookups, the scheme requires only $O(\log^2 n)$ blocks of communication
bandwidth, which constituted roughly an $O(\log n)$-multiplicative improvement
in communication bandwidth over the generic ORAM solution. We will briefly
explain ``the pointer-based technique" they introduced to eliminate the position
map in Section~\ref{sec:ourwork}. 

The practicality of oblivious data structures are challenging, however,
owing to the combination of inefficiencies in the data structures
compounded with that of the underlying ORAM.
In our experimental results presented in \cref{sec:eval}, and
\cref{tab:comm} specifically, we found that the AVL ODS suffers greatly
from a high round complexity, and also that the per-operation bandwidth
exceeds the total database size (and hence a trivial alternative
implementation) until the number of entries exceeds 1 million.

Similar observations for ORAMs more generally were made recently by
Bindschaedler et
al.~\cite{CCS:BNP+15}, who examined existing ORAM alternatives in a realistic
cloud setting, and found many theoretical results lacking in practice.
We ask a related question for ODS, and answer it in the affirmative
with our HIRB data structure stored in vORAM:

\begin{quotation}
  \noindent\em
  Can an oblivious map data structure be made practically useful
  in the cloud setting?
\end{quotation}

\paragraph{Catastrophic attack.}
In the cloud storage scenario, obliviousness will
protect the client's privacy from any observer of network traffic or
from the cloud server itself.
However, if the attacker compromises the client and obtains
critical information such as the encryption keys used in the ODS, all
the sensitive information stored in the cloud will simply be revealed
to the attacker. 

We call this scenario a {\em catastrophic attack}, and it is important
to stress that this attack is quite realistic. The client machine may
be stolen or hacked, or it may even be legally seized due to a
subpoena. 

Considering the increasing incidence of high-profile catastrophic
attacks in practice
(e.g., \cite{photoleak, lastpass}),
and that even government
agencies such the CIA are turning to third-party cloud storage
providers~\cite{cia}, it is important to provide some level of privacy
in this attack scenario. Given this reality, we ask and answer the
following additional question:

\begin{quotation}
\noindent \em Can we provide any privacy guarantee even under a catastrophic
attack?
\end{quotation}

Specifically, our vORAM+HIRB construction will provide strong security
for deleted data, as well as a weaker (yet optimal) 
security for the history of past operations, after complete client
compromise.

\subsection{Security Requirements}
Motivated by the goals outlined previously, we aim to construct a
cloud database system that provides the following two security
properties:
\begin{itemize}
\item Obliviousness: The system should hide both the data and the
    access patterns from an observer of all 
    client-server communication (i.e., be an ODS).
\item Secure Deletion and History Independence: The system, in the
    face of a catastrophic attack, should ensure that no previously
    deleted data, the fact that previous data existed, or the order in
    which extant data has been accessed, is revealed to an attacker.
\end{itemize}
Additionally, we require that the system be practically useful, meaning
it should be {\em more efficient
(w.r.t.~communication cost, access time, and round complexity) than 
previous ODS schemes}, even those
that do not necessarily provide secure deletion nor history independence.

Each required security notion has individually been the focus of
numerous recent research efforts (see Section~\ref{sec:relwork}). To
the best of our knowledge, however, {\em there is no previous work
  that considers all the properties simultaneously.} We aim at
combining the security properties from \emph{obliviousness},
\emph{secure deletion}, and \emph{history independence} into a new,
unified system for secure remote cloud storage. The previous ODS
schemes do not provide history-independence nor secure deletion 
and are inefficient for even small data stores. Previous mechanisms
providing secure deletion or history independence are more efficient,
but do not hide the access pattern in remote cloud storage (i.e., do not
provide obliviousness). And unfortunately, the specific requirements of
these constructions means they cannot trivially be combined in a
straightforward way.

To better understand the necessity of each of the security
requirements, consider each in kind.

\BDS
\item[\it Obliviousness:] The network traffic to a remote server reveals
to an attacker, or to the server itself, which raw blocks are being
read and written. Even if the block contents are encrypted, an attacker
may be able to infer sensitive information from this access pattern itself. 
Like previous ODS schemes, our system will ensure this
is not the case; the server-level access pattern reveals nothing about
the underlying data operations that the user is performing.

\item[\it History independence:] By inspecting the internal structure of the
currently existing data in the cloud after a catastrophic attack, the attacker
may still be able to infer information about which items were recently accessed
or the likely prior existence of a record even if that record was previously
deleted~\cite{ICDE:BajSio13}. However, if an ODS scheme provided {\em perfect
history independence}, the catastrophic attacker cannot infer which sequence of
operations was applied, among all the sequences that could have resulted in the
current set of the data items.
Interestingly, we show that it is impossible to achieve perfect
history independence in our setting with a computationally-limited
server;
nonetheless, providing $\ell$-history independence is still desirable, where
only the most recent $\ell$ operations are revealed but nothing else.

\item[\it Secure deletion:] Given that only bounded history independence is possible,
the privacy of deleted data must be considered. It is desirable that the
catastrophic attacker should not be able to guess information about deleted
data. In practice, data deleted from persistent media, such as hard disk drives,
is easily recoverable through standard forensic tools. In the cloud setting, the
problem is compounded because there is normally no direct control of how and
where data is stored on physical disks, or backed up and duplicated in servers
around the globe. We follow a similar approach as \cite{CCS:RRBC13}, where
secure deletion is accomplished by re-encrypting and deleting the old encryption
key from local, erasable memory such as RAM.
\EDS

\subsection{Our Work}
\label{sec:ourwork}
\paragraph{Pointer-based technique.}
Wang et al.~\cite{CCS:WNLCSS14} designed an ODS scheme for map by storing an AVL
tree on top of the non-recursive Path ORAM~\cite{CCS:SDSFRY13} using the
pointer-based technique, in which the ORAM position tags act as pointers, and
the pointer to each node in the AVL tree is stored in its parent node. With this
technique, when the parent node is fetched, the position tags of its children
are immediately obtained.  Therefore, {\em the position map lookups are no more
necessary}. 

Similarly, in our ODS scheme, we will overlay a data structure on a
non-recursive ORAM using a pointer-based technique for building the data
structure.

We stress that the non-recursive Path ORAM still remains {the best choice} when
we would like to embed our data structure in an ORAM with the pointer-based
technique, in spite of all the recent improvements on ORAM techniques. This is
mainly because all ORAM improvement techniques consider the setting where an
ORAM runs in a stand-alone fashion, unlike our setting where the ORAM actions,
in particular with position map lookups, depend on the upper-layer data
structure. 
In particular, with the non-recursive Path ORAM, each ORAM operation takes only
a {\em single round} of communication between the client and server, since there
is no position map lookup; moreover, each operation transfers $O(\log n)$ blocks
where the size of each block can be arbitrarily small up to $\Omega(\log n)$.
To compare the non-recursive Path ORAM with the most recent stand-alone ORAMs,
each operation of the constant communication ORAM~\cite{MMB15} transfers $O(1)$
blocks each of which should be of size $\Omega(\log^4 n)$, and it additionally
uses computation-intensive homomorphic encryptions. For Ring
ORAM~\cite{RFK+15}, it still refers to the position map, and although its
online stage may be comparable to the non-recursive Path ORAM, it still has the
additional offline stage. The non-recursive version of these ORAMs has
essentially the same efficiency as the non-recursive Path ORAM. 

\paragraph{Impracticality of existing data structures.} 
Unfortunately, no current data structure exists that can meet our
security and efficiency requirements:

\begin{itemize}
\item It should be a {\em rooted tree.} This is necessary, since we
  would like to use the pointer-based technique. 
  Because the positions are randomly re-selected on
  any access to that node, the tree structure is important in order to
  avoid dangling references to old pointers.

\item The height of the tree should be $O(\log n)$ {\em in the worst
    case}.  To achieve obliviousness, all operations must execute with
  the same running time, which implies all operations will be
  padded to some upper bound that is dependent on the height of the
  tree.

\item The data structure itself should be {\em (strongly)
  history-independent},
  meaning the organization of nodes depends only on the current
  contents, and not the order of operations which led to the current
  state.
  As a negative example, consider an AVL tree, which is not history
  independent. Inserting the records A, B, C,
  D in that order; or B, C, D, A in that order; or
  A, B, C, D, E and then deleting E; will each result in a different
  state of the data structure, thereby revealing (under a
  catastrophic attack) information on
  the insertion order and previous deletions. 
\end{itemize}

To the best of our knowledge, there is no data structure satisfying all of the
above conditions. Most tree-based solutions, including AVL trees and
B-trees, are not history independent. Treaps and B-treaps are rooted trees with
history independence, but they have linear height in the worst case. Skip-lists
and B-Skip-lists are history independent and tree-like, but technically they are
not rooted trees and thereby not amenable to the pointer-based technique. That
is, Skip-lists and B-Skip-lists have multiple incoming links, requiring linear
updates in the ORAM to maintain the pointers and position tags in the worst
case. 

\paragraph{HIRB.}
We developed a new data structure, called a HIRB tree (history
independent, randomized B-tree), that satisfies all the aforementioned
requirements.  Conceptually, it is a {\em fixed height} B-tree such that when
each item is inserted, the level in HIRB tree is determined by $\log_\beta n$
trials of (pseudorandom) biased coin flipping where $\beta$ is the block factor.
The tree may split or merge depending on the situation, but {\em it
  never rotates}.
The fixed height of the tree, i.e. $H = 1 + \log_\beta n$, is
very beneficial for efficiency. In particular, {\em every operation}
visits at most $2H$ nodes, which greatly saves on padding costs, compared to
the ODS scheme of~\cite{CCS:WNLCSS14} where each AVL tree operation must be
padded up to visiting $3\cdot 1.44 \cdot \lg n$ nodes. 
 
The HIRB is described more carefully in \cref{sec:ds}, with full details
in the appendix.

\paragraph{vORAM.}
One challenge with HIRB trees is that number of items that each tree
node contains are variable, and in the unlucky case, it may become too
large for an ORAM bucket to store. 

This challenge is overcome by
introducing vORAM (ORAM with variable-size blocks). The design of
vORAM is based on the non-recursive version of Path ORAM where the
bucket size remains fixed, but each bucket may contain as many
variable-size blocks (or parts of blocks) as the bucket space allows. Blocks
may also be stored across multiple buckets (in the same path). 

We observe that
the irregularity of the HIRB node sizes can be smoothed over $O(\log n)$ buckets
from the vORAM root to an vORAM leaf, and we prove that the stash size on the
client can still be small $\tilde O(\log n)$ with high probability. 
We note that vORAM is {\em the
first ORAM that deals with variable size blocks}, 
and may be of independent interest. 

The vORAM is described carefully in \cref{sec:oram}, and
the full details are provided in the appendix.

\paragraph{Secure deletion.}
Finally, for secure deletion, a parent vORAM bucket contains the encryption keys
of both children. When a bucket is modified, it is encrypted with a fresh key;
then the encryption keys in the parent is accordingly modified, which
recursively affects all its ancestors. 
However, we stress that in each vORAM operation, {\em leaf-to-root refreshing
takes place anyway, and adding this mechanism is bandwidth-free.}
Additionally, instead of using the label of each item directly in HIRB, we use
the hash of the label.  This way, we can remove the dependency between the item
location in HIRB and its label (with security proven in the random oracle
model).

\paragraph{Imperfect history independence.} 
Our approach does not provide perfect history independence. Although
the data structure in the vORAM is history independent, the vORAM is
not. Indeed, in any tree-based or hierarchical ORAM, the items near
the root have been more likely recently accessed as compared to items
near the leaves. The catastrophic adversary can observe all the ORAM
structure, and such leakage breaks perfect history independence. We
show a formal lower bound for the amount of leakage in 
Section~\ref{sec:prelim}.

\paragraph{Experiments and efficiency of our scheme.}
In order to empirically measure the performance of our construction,
we first performed an analysis to determine the smallest constant
factor overhead to achieve high performance with negligible likelihood of
failure. Following this, we implemented our system in the cloud with Amazon Web
Services as the cloud provider and compared it to alternatives that
provide some, but not all of the desired security properties.
To the best of our knowledge, {\em there has been no previous work that implements
and tests any ODS system in the actual cloud setting}. As argued in
Bindschaedler et al.~\cite{CCS:BNP+15}, who independently compared various ORAM
systems in the cloud, it is important to see how systems work in the actual
intended setting. 
As comparison points, we compare our system with the following implementations: 
\BI
\item ORAM+AVL: We reimplemented the ODS map by Wang et al.~\cite{CCS:WNLCSS14}
    that provides obliviousness but not secure deletion nor history independence.
\item SD-B-Tree: We implemented a remotely stored
block-level, encrypted B-Tree (as recommend by the secure deletion
community~\cite{CCS:RRBC13}) that provides secure deletion but not history
independence nor obliviousness.  
\item Naive approach: We implemented a naive approach that achieves all the
    security properties by transferring and re-encrypting the entire database on
    each access.
\EI

In all cases the remotely stored B-Tree is the fastest as it requires
the least amount of communication cost (no obliviousness). For similar
reasons, vORAM+HIRB is much faster than the baseline as the number of
items grows (starting from $2^{14}$ items), since the baseline
requires communication that is linear in the number of items. We also
describe a number of optimizations (such as concurrent connections and
caching) that enables vORAM+HIRB to be competitive with the baseline
even when storing as few as $2^9$ items. It should be noted, without
optimizations, the access time is on the order of a few seconds, and
with optimizations, access times are less than one second.

Surprisingly, however, the vORAM+HIRB is {\em 20x faster than
  ORAM+AVL}, irrespective of the number of items, even though ORAM+AVL
does not support history independence or secure deletion. We believe
this is mainly because vORAM+HIRB requires much smaller round complexity. Two
factors drive the round complexity improvement:
\begin{description}
\item[\it Much smaller height:] While each AVL tree node contains only
  one item, each HIRB node contains $\beta$ items on average, and is
  able to take advantage of slightly larger buckets which optimize the
  bandwidth to remote cloud storage by storing the same amount of data in
  trees with smaller height.

\item[\it Much less padding:] AVL tree operations sometimes get complicated
    with balancing and rotations, due to which each operation should be padded
    up to $3 \cdot 1.44 \lg n$ node accesses. However, HIRB operations are simple,
    do not require rotations, and thus, each operation accesses at most $2
    \log_\beta n$ nodes.
\end{description}
Although the Path-ORAM bucket for ORAM+AVL is four times smaller than the vORAM bucket
in our implementation, it affects bandwidth but not the round complexity.  The fully
optimized vORAM+HIRB protocol is about {\em 100x faster than ORAM+AVL}. We
describe the details of our experiments in Section~\ref{sec:eval}.

\paragraph{Summary of our contributions.}
To summarize, the contributions of this paper are:
\begin{itemize}[noitemsep]
\item New security definitions of history independence and secure deletion under
    a catastrophic attack. 
\item The design and analysis of an oblivious RAM with variable size
  blocks, the vORAM;
\item The design and analysis of a new history independent and
  randomized data structure, the HIRB tree;
\item A lower bound on history independence for any
  ORAM construction with sub-linear bandwidth;
\item Improvements to the performance of mapped data structures stored
  in ORAMs;
\item An empirical measurement of the settings and performance of the
  vORAM in the actual cloud setting;
\item The implementation and measurement of the vORAM+HIRB
  system in the actual cloud setting. 
\end{itemize}

\section{Related Work}
\label{sec:relwork}

We discuss related work in oblivious data structures, history independence, and
secure deletion. Our system builds upon these prior results and combines the
security properties into a unified system.

\paragraph{ORAM and oblivious data structures.}
ORAM protects the access pattern from an observer such that it
is impossible to determine which operation is occurring, and on which
item. The seminal work on the topic
is by Goldreich and Ostrovsky~\cite{GO96}, and since then, many works have
focused on improving efficiency of ORAM in both the space, time, and
communication cost complexities (for example
\cite{TCC:DamMelNie11,SODA:GMOT12,SODA:KusLuOst12,NDSS:SSS12,CCS:SDSFRY13,usenix:DSS14}
just to name a few; see the references therein).
 
There have been works addressing individual oblivious data structures to
accomplish specific tasks, such as priority queues~\cite{PODC:Toft11}, stacks
and queues~\cite{STACS:MZ14}, and graph algorithms~\cite{ASIACCS:BlaSteAli13}.
Recently, Wang et al.~\cite{CCS:WNLCSS14} achieved oblivious data structures
(ODS) for maps, priority queues, stacks, and queues much more efficiently than
previous works or naive implementation of the data structures on top of
ORAM. 

Our vORAM construction builds upon the non-recursive Path
ORAM~\cite{CCS:WNLCSS14} and allows \emph{variable sized} data items to be
spread across multiple ORAM buckets.  
Although our original motivation was to store differing-sized B-tree
nodes from the HIRB, there may be wider applicability to any context
where the size (as well as contents and access patterns) to data needs
to be hidden. 

Interestingly, based on our experimental results, we believe the ability
of vORAM to store partial blocks in each bucket may even improve the
performance of ORAM when storing uniformly-sized items. However, we will
not consider this further in the current investigation.

\paragraph{History independence.} 
History independence of data structures requires that the current organization
of the data within the structure reveals nothing about the prior operations
thereon. 
Micciancio~\cite{STOC:Micciancio97} first considered history independence in the
context of 2-3 trees, and the notions of history independence were formally
developed in~\cite{STOC:NaoTea01,Algo:HHMPR05,IC:BP06}.  The notion of
\emph{strong history independence}~\cite{STOC:NaoTea01} holds 
if for any two sequences of
operations, the distributions of the memory representations are
identical at all time-points
that yield the same storage content. Moreover, a data structure is
strongly history independent if and only if it has a {\em unique
representation}~\cite{Algo:HHMPR05}. There have been uniquely-represented
constructions for hash functions~\cite{FOCS:BleGol07,ICALP:NaoSegWie08} and
variants of a B-tree (a B-treap~\cite{ICALP:Golovin09}, and a
B-skip-list~\cite{arxiv:Gol10}). 
We adopt the notion of unique representation for history independence when
developing our history independent, randomized B-tree, or HIRB tree.

We note that history independence of these data structures considers a setting
where a single party runs some algorithms on a single storage medium, which
doesn't correctly capture the actual cloud setting where client and server have
separate storage, execute protocols, and exchange messages to maintain the data
structures. Therefore, we extend the existing history independence and give a
new, augmented notion of history independence for the cloud setting with a
catastrophic attack. 

Independently, the recent work of \cite{arxiv:BCS15} also considers a limited
notion of history independence, called $\Delta$-history independence,
parameterized with a function $\Delta$ that describes the leakage. 
Our definition of history independence has a similar notion, where the
leakage function $\Delta$ captures the number of recent operations which
may be revealed in a catastrophic attack.

\paragraph{Secure deletion.}
Secure deletion means that data deleted cannot be recovered, even by the
original owner. It has been studied in many contexts \cite{SP:RBC13},
but here we focus on the cloud setting, where the user has little or no
control over the physical media or redundant duplication or backup
copies of data.
In particular, we build upon secure deletion techniques from the applied
cryptography community.  
The approach is to encrypt all data stored in the cloud with encryption
keys stored locally in erasable memory, so that deleting the keys will
securely delete the remote data by rendering it non-decryptable.

Boneh and Lipton~\cite{usenix:BL96} were the first to use encryption
to securely remove files in a system with backup tapes. The challenge
since was to more effectively manage encrypted content and the
processes of re-encryption and erasing decryption keys. For example,
Di Crescenzo et al.~\cite{STACS:DFIJ99} showed a more efficient method
for secure deletion using a tree structure applied in the setting of a
large non-erasable persistent medium and a small erasable medium.
Several works considered secure deletion mechanisms for a versioning
file system~\cite{FAST:PBHSR05}, an inverted index in a
write-once-read-many compliance storage~\cite{EDBT:MWB08}, and a
B-tree (and generally a mangrove)~\cite{CCS:RRBC13}.

\section{Preliminaries}
\label{sec:prelim}

We assume that readers are familiar with security notions of
standard cryptographic primitives~\cite{KatzLindell2007}. 
Let $\secparam$ denote the security parameter. 

\paragraph{Modeling data structures.}
Following the approach from the secure deletion literature, we use
two storage types: {\em erasable memory} and {\em persistent storage}.
Contents deleted from erasable memory are non-recoverable, while the
contents in persistent storage cannot be fully erased.  We assume 
the size of erasable memory is small while the persistent storage has a much
larger capacity. This mimics the cloud computing setting where cloud storage is
large and persistent due to lack of user control, and local storage is
more expensive but also controlled directly.

We define a data structure $\DDD$ as a collection of data that supports
initialization, insertion, deletion, and lookup, using both the erasable memory
and the persistent storage. Each operation may be parameterized by some operands
(e.g., lookup by a label).
For a data structure $\DDD$ stored in this model, let $\DDD.\mem$ and
$\DDD.\hdd$ denote the contents of the erasable memory and persistent storage,
respectively. 
For example, an encrypted graph structure may be stored in $\DDD.\hdd$ while the
decryption key resides in $\DDD.\mem$. 
For an operation $\op$ on $\DDD$, let $\acc \from \DDD.\op()$ denote executing
the operation $\op$ on the data structure $\DDD$ where $\acc$ is the access
pattern over the persistent storage during the operation. The access pattern to
erasable memory is assumed to be hidden. 
For a sequence of operations $\Vec{\op} = (\op_1, \ldots, \op_m)$, let
$\Vec{\acc} \from \DDD.\Vec{\op}()$ denote applying the operations on
$\DDD$, that is,
$\acc_1 \from \DDD.\op_1(), ~~\ldots~~, \acc_m \from \DDD.\op_m(),$
with $\Vec{\acc} = (\acc_1, \ldots, \acc_m)$.
We note that the access pattern $\Vec{\acc}$ completely determines the state of
persistent storage $\DDD.\hdd$.

\paragraph{Obliviousness and history independence.}
Obliviousness requires that the adversary without access to erasable memory
cannot obtain any information about actual operations performed on data
structure $\DDD$ other than the number of operations.  This security notion is
defined through an experiment \obh, given in Figure~\ref{fig:secdef}, where
$\DDD$, $\secparam$, $n$, $h$, $b$ denote a data structure, the security
parameter, the maximum number of items $\DDD$ can contain, history independence,
and the challenge choice. 

In the experiment, the adversary
chooses two sequences of operations on the data structure and tries to guess
which sequence was chosen by the experiment with the help of access patterns.
The data structure provides obliviousness if every polynomial-time adversary has
only a negligible advantage.

\begin{figure}[t]
\centering
\begin{minipage}[t]{\minof{6 cm}{0.52\linewidth}}
\small
$\expobh$  \\
$\acc_0 \from \DDD.\init(1^\secparam, n);$ \\
$(\Vec{\op}^{(0)}, \Vec{\op}^{(1)}, \st) \from \AAA_1(1^\secparam,\acc_0);$ \\
$\Vec{\acc} \from \DDD.\Vec{\op}^{(b)}()$; \\
if $h=1$: \\
\phantom{x} return $\AAA_2(\st, \Vec{\acc}, \DDD.\mem);$ \\
else \\
\phantom{x} return $\AAA_2(\st, \Vec{\acc});$ \\
\end{minipage}%
\begin{minipage}[t]{\minof{6 cm}{0.48\linewidth}}
\small
$\expsdel$  \\
$\acc_0 \gets \DDD.\init(1^\secparam, n);$ \\
$d_0 \gets \AAA_1(1^\secparam, 0)$; \\
$d_1 \gets \AAA_1(1^\secparam, 1)$; \\
$(\Vec{\op}_{d_0, d_1}, S) \from \AAA_2(\acc_0, d_0, d_1)$; \\
$\Vec{\acc} \from \DDD.(\Vec{\op}_{d_0, d_1} \Join_S d_b)()$; \\ return
$\AAA_3(\acc_0, \Vec{\acc}, \DDD.\mem);$  
\end{minipage}
\caption{Experiments for security definitions \label{fig:secdef}}
\end{figure}

\BD For a data structure $\DDD$, consider the experiment $\expob$ with
adversary $\AAA = (\AAA_1, \AAA_2)$. 
We call the adversary $\AAA$ $\sf admissible$ if $\AAA_1$ always outputs two
sequences with the same number of operations storing at most $n$ items. 
We define the advantage of the adversary $\AAA$ in this experiment as: $$
\advob = \left | \begin{array}{l} \Pr[\expobz=1]  \\ ~~ -
\Pr[\expobo=1]\end{array} \right |.$$
We say that $\DDD$ provides $\sf obliviousness$ if for any sufficiently large
$\secparam$, any $n \in poly(\secparam)$, and any PPT admissible adversary
$\AAA$, we have $\advob \le \negl(\secparam).$
\ED

Now we define history independence.  As we will see, perfect history
independence is inherently at odds with obliviousness and sub-linear
communication cost. Therefore, we define \emph{parameterized} history
independence instead that allows for a relaxation of the security requirement.
The parameter determines the allowable leakage of recent history of operations.
One can interpret a history-independent data structure with leakage of $\ell$
operations as follows: Although the data structure may reveal some recent $\ell$
operations applied to itself, it does not reveal any information about older
operations, except that the total sequence resulted in the current state of data
storage. 

The experiment in this case is equivalent to that for obliviousness,
except that (1) the two sequences must result in the same state of the
data structure at the end, (2) the last $\ell$ operations in both
sequences must be identical, and (3) the adversary gets to view the
local, erasable memory as well as the access pattern to persistent
storage.

\BD For a data structure $\DDD$, consider the experiment $\exphi$ with
adversary $\AAA = (\AAA_1, \AAA_2)$. 
We call the adversary $\AAA$ $\sf \ell\mbox{-}admissible$ if $\AAA_1$ always
outputs sequences $\Vec{\op}^{(0)}$ and $\Vec{\op}^{(1)}$ which have the same
number of operations and result in the same set storing at most $n$ data items,
and the last $\ell$ operations of both are identical. 
We define the advantage of an adversary $\AAA$ in this experiment above as: $$
\advhi = \left | \begin{array}{l} \Pr[\exphiz=1]  \\ ~~ -
\Pr[\exphio=1]\end{array} \right |.$$
We say that the data structure $\DDD$ provides $\sf history~independence$ with
leakage of $\ell$ operations  if for any sufficiently large $\secparam$, any $n
\in poly(\secparam)$, and any PPT $\ell$-admissible adversary $\AAA$, we have
$\advhi \le \negl(\secparam).$ 
\ED

\paragraph{Lower bound on history independence.}
Unfortunately, the history independence property is inherently at odds with the
nature of oblivious RAM. The following lower bound demonstrates that there is a
linear tradeoff between the amount of history independence and the communication
bandwidth of any ORAM mechanism.

\begin{restatable}{theorem}{oramhilb}\label{thm:oramhilb}
  Any oblivious RAM storage system with a bandwidth of $k$ bytes per
  access achieves at best history independence with leakage of
  $\Omega(n/k)$ operations in
  storing $n$ blocks.
\end{restatable}

The intuition behind the proof%
\footnote{Full proofs for the main theorems may be found in Appendix~\ref{app:proofs}.}
is that, in a catastrophic attack, an
adversary can observe which persistent storage locations were recently
accessed, and furthermore can decrypt the contents of those locations
because they have the keys from erasable memory. This will inevitably
reveal information to the attacker about the order and contents of
recent accesses, up to the point at which all $n$ elements have been
touched by the ORAM and no further past information is recoverable.

Admittedly this lower bound limits what may be achievable in terms of
history independence. But still, leaking only a known maximum number of
prior operations is better than (potentially) leaking all of them!

Consider, by contrast, an AVL tree implemented within a standard ORAM as
in prior work.
Using the fact that AVL tree shapes reveal
information about past operations, the adversary can come up with two sequences
of operations such that (i) the first operations of each sequence result in a
distinct AVL tree shape but the same data items, and (ii) the same read
operations, as many as necessary, follow at the end.  With the catastrophic
attack, the adversary will simply observe the tree shape and make a correct
guess. This argument holds for any data structure whose shape reveals
information about past operations, which therefore have \emph{no upper
bound} on the amount of history leakage.

\paragraph{Secure deletion.} 
Perfect history independence implies secure deletion. However, the above 
lower bound shows that complete history independence will not be possible in our
setting. So, we consider a complementary security notion that requires {\em
strong security for the deleted data}. 
Secure deletion is defined through an experiment \sdel, given in
Figure~\ref{fig:secdef}. In the experiment, $\AAA_1$ chooses two data items
$d_0$ and $d_1$ at random, based on which $\AAA_2$ outputs $(\Vec{\op}_{d_0, d_1}, S)$. Here,
$\Vec{\op}_{d_0, d_1}$ denotes a vector of operations containing neither $d_0$ nor
$d_1$, and $S = (s_1, s_2, \ldots, s_m)$ is a monotonically increasing sequence.
$\Vec{\op}_{d_0,d_1} \Join_S d_b$ denotes injecting $d_b$ into $\Vec{\op}_{d_0,
d_1}$ according to $S$. In particular, ``insert $d_b$" is placed at position
$s_1$; for example, if $s_1$ is 5, this insert operation is placed right before
the 6th operation of $\Vec{\op}_{d_0,d_1}$. Then, ``look-up $d_b$" is placed at
positions $s_2, \ldots, s_{m-1}$, and finally ``delete $d_b$" at $s_m$. 

\BD For a data structure $\DDD$, consider the experiment $\expsdel$
with adversary $\AAA = (\AAA_1, \AAA_2, \AAA_3)$. 
We call the adversary $\AAA$ $\sf admissible$ if for any data item $d$ that
$\AAA_1(1^\secparam, 0)$ (resp., $\AAA_1(1^\secparam, 1)$) outputs, the
probability that $\AAA_1$ outputs $d$ is negligible in $\secparam$, i.e., the
output $\AAA_1$ forms a high-entropy distribution; moreover, the sequence of
operations from $\AAA_2$ must store at most $n$ items. 
We define the advantage of $\AAA$ as: $$ \advsdel = \left | \begin{array}{l}
\Pr[\expsdelz=1]  \\ ~~ - \Pr[\expsdelo=1]\end{array} \right |.$$
We say that the data structure $\DDD$ provides $\sf secure~deletion$ if for any
sufficiently large $\secparam$, any $n \in poly(\secparam)$, and any PPT
admissible adversary $\AAA$, we have $\advsdel \le \negl(\secparam).$
\ED

Note that our definition is stronger than just requiring that the adversary
cannot recover the deleted item; for any two high entropy distributions chosen
by the adversary, the adversary cannot tell from which distribution the deleted item
was drawn.

\section{ORAM with variable-size blocks (vORAM)}
\label{sec:oram}

The design of vORAM is based on the non-recursive version of Path
ORAM~\cite{CCS:SDSFRY13}, but we are able to add more flexibility by allowing
each ORAM bucket to contain as many variable-size blocks (or parts of blocks) as
the bucket space allows. We will show that vORAM preserves obliviousness and
maintains a small stash as long as the size of variable blocks can be bounded by
a geometric probability distribution, which is the case for the HIRB that we
intend to store within the vORAM. To support secure deletion, we also store
encryption keys within each bucket for its two children, and these keys are
re-generated on every access, similarly to other work on secure deletion
\cite{STACS:DFIJ99,CCS:RRBC13}.

\paragraph{Parameters.} The vORAM construction is governed by the following
parameters: 

\BI
\item 
The height $T$ of the vORAM tree: The vORAM is represented as a complete binary
tree of buckets with height $T$ (the levels of the tree are numbered 0 to $T$),
so the total number of buckets is $2^{T+1} - 1$.  $T$ also controls the total
number of allowable data blocks, which is $2^{T}$. 

\item 
The bucket size $Z$: Each bucket has $Z$ bits, and this $Z$ must be at least
some constant times the \emph{expected block size} $B$ for what will be stored
in the vORAM.

\item
The stash size parameter $R$: Blocks (or partial blocks) that overflow from the
root bucket are stored temporarily in an additional memory bank in local
storage called the stash, which can contain up to $R \cdot B$ bits. 

\item 
Block collision parameter $\hashparam$: Each block will be assigned a
random identifier $id$; these identifiers will all be distinct at every
step with
probability $1 - \negl(\hashparam)$. 
\EI

\paragraph{Bucket structure.} Each bucket is split into two areas: header and
data. See Figure \ref{fig:bucket} for a pictorial description.  The header area
contains two encryption keys for the two child buckets.  The data area contains
a sequence of (possibly partial) blocks, each preceded by a unique identifier
string and the block data length. The end of the data area is filled with 0
bits, if necessary, to pad up to the bucket size $Z$.

\begin{figure}[t]
\begin{center}
\noindent\begin{tikzpicture}[node distance=0]

\def\inode#1#2{%
\node[short,label=center:{#2}] (#1) {}}
\def\shnode#1#2#3{%
\node[short,right=of #1, label=center:{#3}] (#2) {}}
\def\lnode#1#2{%
  \node[long,right=of #1] (#2) {}}
\def\llnode#1#2#3{%
\node[long,right=of #1, label=center:{#3}] (#2) {}}
\def\mnode#1#2#3{%
\node[long,right=of #1, label=center:{#3}] (#2) {}}

\inode{a}{$k_1$};
\shnode{a}{b}{$k_2$};
\shnode{b}{c}{$id_1$};
\shnode{c}{d}{$l_1$};
\llnode{d}{e}{$blk_1$};
\node[short,draw=none,fill=none,right=of e,text height=0pt,text depth=0pt,text
width=0.5cm] (g) {$\ldots$};
\shnode{g}{h}{$id_\ell$};
\shnode{h}{i}{$l_\ell$};
\mnode{i}{j}{$blk_\ell$};
\shnode{j}{k}{$0$};

\draw[decorate,decoration={brace,raise=2pt, mirror}] (a.south west) -- node[below=4pt]
{header} (b.south east);

\draw[decorate,decoration={brace,raise=2pt,mirror}] (e.south west) -- node[below=4pt]
{$l_1$ bytes} (e.south east);

\draw[decorate,decoration={brace,raise=2pt,mirror}] (j.south west) -- node[below=4pt]
{$l_\ell$ bytes} (j.south east);
\end{tikzpicture}
\end{center}
\vspace{-1em}
\caption{A single vORAM bucket with $\ell$ partial blocks\label{fig:bucket}.}
\vspace{-1em}
\end{figure}
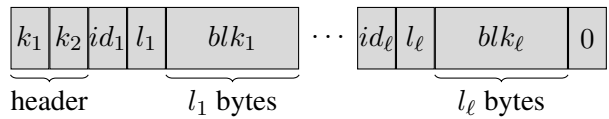

\visoram
Each $id_i$ uniquely identifies a block and also encodes the path of buckets
along which the block should reside.  Partial blocks share the same
identifier with each length $l$ 
indicating how many bytes of the block are stored in
that bucket. Recovering the full block is accomplished by scanning from the
stash along the path associated with $id$ (see Figure~\ref{fig:voram-vis}). We
further require the first bit of each identifier to be always 1 in order to
differentiate between zero padding and the start of next identifier. Moreover,
to avoid collisions in identifiers, the length of each identifier is extended to
$2T+\gamma+1$ bits, where $\gamma$ is the collision parameter mentioned above.
The most significant $T+1$ bits of the identifier (including the fixed leading
1-bit) are used to match a block to a leaf, or equivalently, a path from root to
leaf in the vORAM tree.

\paragraph{vORAM operations.}
Our vORAM construction supports the following operations. 
\BI
\item $\ins(blk) \mapsto id$. Inserts the given block $blk$ of data into the
    ORAM and returns a new, randomly-generated $id$ to be used only once
    at a later time to retrieve the original contents.

\item $\remove(id) \mapsto blk$. Removes the block corresponding to $id$ and
    returns the original data $blk$ as a sequence of bytes.  

\item $\upd(id, callback) \mapsto id^+$. Given $id$ and a user-defined function
    $callback$, perform $\ins( callback( \remove(id)))$ in a single step. 
\EI

Each vORAM operation involves two phases:

\BEN
\item $\evict(id)$. Decrypt and read the buckets along the path from the
    root to the leaf encoded in the identifier $id$, and remove 
    all the partial blocks along the path, merging partial blocks
    that share an identifier, and storing them in the stash.
    
  \item $\writeback(id)$. Encrypt all blocks along the path encoded by
    $id$ with new encryption keys and
    opportunistically store any partial blocks from stash,
    dividing blocks as necessary, filling from the leaf to the root.
\EEN

An \ins{} operation first evicts a randomly-chosen path, then inserts
the new data item into the stash with a second randomly-chosen
identifier, and finally writes back the originally-evicted path. A
\remove{} operation evicts the path specified by the identifier, then
removes that item from the stash (which must have had all its partial
blocks recombined along the evicted path), and finally writes back the
evicted path without the deleted item.
The \upd{} operation evicts the path from the initial $id$, retrieves the
block from stash, passes it to the $callback$ function, 
re-inserts the result to the stash with a new random $id^+$,
and finally
calls \writeback{} on the original $id$.
A full pseudocode description of all these operations is provided in
Appendix~\ref{app:voram}.

\paragraph{Security properties.}
For obliviousness, any $\ins$, $\remove$, $\upd$ operation is
computationally indistinguishable based on its access pattern because
the identifier of each block is used only once to retrieve that item
and then immediately discarded. Each $\remove$ or $\upd$ trivially discards the
identifier after reading the path, and each $\ins$ evicts buckets along a bogus,
randomly chosen path before returning a fresh $id^+$ to be used as the new
identifier for that block.

\begin{restatable}{theorem}{oramob} \label{thm:oramob}
The vORAM provides obliviousness.
\end{restatable}

Secure deletion is achieved via key management of buckets. Every
$\evict$ and $\writeback$ will result in a path's worth of buckets to be
re-encrypted and re-keyed, including the root bucket. Buckets
containing any removed data may persist, but the decryption keys are
erased since the root bucket is re-encrypted, rendering the data
unrecoverable. Similarly, recovering any previously deleted data reduces to
acquiring the old-root key, which was securely deleted from local,
erasable memory. 

However, each $\evict$ and $\writeback$ will disclose the vORAM path
being accessed, which must be handled carefully to ensure no leakage
occurs.  Fortunately, identifiers (and therefore vORAM paths as well)
are uniformly random, independent of the deleted data and revealing no
information about them.

\begin{restatable}{theorem}{oramsd} \label{thm:oramsd}
The vORAM provides secure deletion.
\end{restatable}

Regarding history independence, although any removed items are
unrecoverable, the height of each item in the vORAM tree, as well as the
history of accesses to each vORAM tree bucket, may reveal some
information about the \emph{order}, or timing, of when each item was
inserted. Intuitively, items appearing closer to the root level of the
vORAM are more likely to have been inserted recently, and vice versa.
However, if an item is inserted and then later has its path entirely
\evict{}ed due to some other item's insertion or removal, then any
history information of the older item is essentially wiped out; it is as
if that item had been removed and re-inserted. Because the identifiers
used in each operation are chosen at random, after some $O(n\log n)$
operations it is likely that \emph{every} path in the vORAM has been
evicted at least once.

\begin{restatable}{theorem}{oramhi} \label{thm:oramhi}
The vORAM provides history independence with
leakage of $O(n \log n + \secparam n)$ operations.
\end{restatable}

In fact, we can achieve asymptotically optimal leakage with only a
constant-factor blowup in the bandwidth. Every vORAM operation
involves reading and writing a single path. Additionally, after each
operation, we can evict and then re-write a complete \emph{subtree} of
size $\lg n$ which contains $(lg n)/2 -1$ leaf buckets in a
deterministicly chosen dummy operation that simply reads the buckets
into stash, then rewrites the buckets with no change in contents but
allowing the blocks evicted from the dummy operation and those evicted
from the access to all move between levels of the vORAM as usual. The
number of nodes evicted will be less than $2\lg n$, to encompass the
subtree itself as well as the path of buckets to the root of the
subtree, and hence the total bandwidth for the operation remains
$O(\log n)$.

The benefit of this approach is that if these dummy subtree evictions
are performed \emph{sequentially} across the vORAM tree on each
operation, any sequence of $n/\lg n$ operations is guaranteed to have
evicted every bucket in the vORAM at least once.  Hence this would
achieve history independence with only $O(n/\log n)$ leakage, which
matches the lower bound of Theorem~\ref{thm:oramhilb} and is therefore
optimal up to constant factors.

\paragraph{Stash size.}
Our vORAM construction maintains a small stash as long as the size of variable
blocks can be bounded by a geometric probability distribution, which is the case
for the HIRB that we intend to store within the vORAM. 
\begin{restatable}{theorem}{oramstash} \label{thm:oramstash}
Consider a vORAM 
with $T$ levels, collision parameter $\hashparam$,
storing at most $n = 2^T$ blocks, where the length $l$ of each block is
chosen independently from a distribution such that $\EE[l] = B$ and 
$\Pr[l > mB] < 0.5^m$.
Then, if the bucket size $Z$ satisfies $Z \ge 20 B$, 
for any $R \ge 1$, and after any single access to the vORAM,
we have
$$\Pr[|\stash| > R\,B] < 28\cdot (0.883)^R.$$
\end{restatable}

Note that the constants $28$ and $0.883$ are technical artifacts
of the analysis, and do not matter except to say that $0.883 < 1$ and
thus the failure probability decreases exponentially with the size of
stash.

As a corollary, for a vORAM storing at most $n$ blocks, 
the cloud storage requirement is $40Bn$ bits,
and the bandwidth for each operation amounts to $40 B \lg n$ bits. 
However, this is a theoretical upper bound, and our experiments in
Section~\ref{sec:eval} show a smaller
constants suffice. namely, setting $Z = 6B$ and $T = \lceil \lg n - 1
\rceil$ stabilizes the stash, so that the actual storage requirement
and bandwidth per operation are $6Bn$ and $12 B \lg n$ bits,
respectively.

Furthermore, to avoid failure due to stash overflow or collisions, 
the client storage $R$ and collision parameter $\hashparam$
should both grow slightly faster than $\log n$, i.e.,
$R,\hashparam \in \omega(\log n)$.

\section{HIRB Tree Data Structure}
\label{sec:ds}

We now use the vORAM construction described
in the previous section to implement a data structure supporting the operations
of a \emph{dictionary} that maps labels to values. In this paper, we
intentionally use the word ``labels'' rather than the word ``keys'' to
distinguish from the encryption \emph{keys} that are stored in the vORAM.

\paragraph{Motivating the HIRB.}
Before describing the construction and properties of the history
independent, randomized B-Tree (HIRB), we first wish to motivate the
need for the HIRB as it relates to the security and efficiently
requirements of storing it within the vORAM:
\begin{itemize}
\item The data structure must be easily partitioned into blocks
  that have expected size bounded by a geometric distribution for
  vORAM storage.

\item The data structure must be \emph{pointer-based}, and the
  structure of blocks and pointers must form a directed graph that is
  an \emph{arborescence}, such that there exists \emph{at most one}
  pointer to each block.
  This is because a non-recursive ORAM uses random identifiers for
  storage blocks, which must change on every read or write to that
  block.

\item The memory access pattern for an operation (e.g., \get{},
  \set{}, or \delete{}) must be bounded by a fix parameter to ensure
  obliviousness; otherwise the number of vORAM accesses could leak
  information about the data access.

\item Finally, the data structure must be \emph{uniquely represented}
  such that the pointer structures and contents are determined only by
  the set of $(\dskey,\dsval)$ pairs stored within, up to some
  randomization performed during initialization. Recall that {\em
    strong history independence} is provided via a unique
  representation, a sufficient and necessary
  condition~\cite{Algo:HHMPR05} for the desired
  security  property.

\end{itemize}

In summary, we require a uniquely-represented, tree-based data structure with
bounded height. While a variety of uniquely represented (or strongly history
independent) data structures have been proposed in the
literature~\cite{STOC:NaoTea01,ICALP:Golovin09}, we are not aware of any that
satisfy all of the requisite properties. 

While some form of hash table might seem like an obvious choice, we note
that such a structure would violate the second condition above; namely,
it would be impossible to store a hash table within an ORAM without
having a separate \emph{position map}, incurring an extra logarithmic
factor in the cost.  As it turns out, our HIRB
tree does use hashing in order to support secure deletion, but this is
only to sort the labels within the tree data structure.

\paragraph{Overview of HIRB tree.}
The closest data structure to the HIRB is the B-Skip List
\citep{arxiv:Gol10}; unfortunately, a skip list does not form a
tree. The HIRB is essentially equivalent to a B-Skip List after
sorting labels according to a hash function and removing pointers
between skip-nodes to impose a top-down tree structure.

Recall that a typical B-tree consists of large nodes, each with an array
of $(\dskey,\dsval)$ pairs and child nodes. A B-tree node has
branching factor of $k$, and we call it a $k$-node, if the node contains
$k-1$ labels, $k-1$ values, and $k$ children (as in
Figure~\ref{fig:knode}). In a typical B-tree, the branching factor
of each node is allowed to vary in some range $[B+1, 2B]$, where $B$
is a fixed parameter of the construction that controls the maximum
size of any single node.

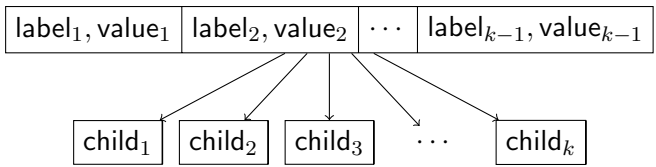
\begin{figure}[htbp]
\begin{center}
\begin{tikzpicture}
    \tikzstyle{hirb}=[rectangle split, rectangle split horizontal,rectangle split ignore empty parts,draw]
    \tikzstyle{every node}=[hirb]
    \tikzstyle{level 1}=[sibling distance=40 pt]

    \node {$\dskey_1, \dsval_1$ 
           \nodepart{two} $\dskey_2, \dsval_2$
           \nodepart{three} $\cdots$
           \nodepart{four} $\dskey_{k-1}, \dsval_{k-1}$
           } [->]
      child {node {$\child_1$}}
      child {node {$\child_2$}}
      child {node {$\child_3$}}
      child {node[draw=none] {$\cdots$}}
      child {node {$\child_k$}}
    ;
\end{tikzpicture}
\vspace{-1em}
\end{center}
\caption{B-tree node with branching factor $k$\label{fig:knode}}
\end{figure}

HIRB tree nodes differ from typical B-tree nodes in two ways. First,
instead of storing the label in the node a cryptographic hash%
\footnote{
We need a random oracle for formal security. In practice, we used a SHA1
initialized with a random string chosen when the HIRB tree is instantiated.} 
of the label is stored.
This is necessary to support secure deletion of vORAM+HIRB
even when the nature of vORAM leaks some history of operations; namely,
revealing which HIRB node an item was deleted from should not reveal
the label that was deleted.

The second difference from a normal B-tree node is that the branching factor
of each node, rather than being limited to a fixed range, can take any
value $k \in [1,\infty)$. This branching factor will observe a
geometric distribution for storage within the vORAM. In particular,
it will be a random variable $X$ drawn
independently from a geometric distribution with \emph{expected} value
$\beta$, where $\beta$ is a parameter of the HIRB tree construction.

The \emph{height} of a node in the HIRB tree is defined as the length
of the path from that node to a leaf node; all leaf nodes are the same
distance to the root node for B-trees. The height of a new insertion of
$(\dskey,\dsval)$ in the HIRB is determined by a series of pseudorandom biased
coin flips based on the hash of the \dskey{}\footnote{Note that this choice of
    heights is more or less the same as the randomly-chosen node heights in a
skip list.}. The distribution of selected heights for insertions uniquely
determines the structure of the HIRB tree because the process is deterministic,
and thus the HIRB is uniquely-represented.

\paragraph{Parameters and preliminaries.}
Two parameters are fixed at initialization: the \emph{expected branching factor}
$\beta$, and the \emph{height} $H$. In addition, throughout this section we will
write $n$ as the maximum number of distinct labels that may be stored in the
HIRB tree, and $\hashparam$ as a parameter that affects the length of
hash digests%
\footnote{The parameter $\hashparam$ for HIRB and vORAM serves the
same purpose in avoiding collisions in identifiers so for simplicity
we assume they are the same}%
.

A HIRB tree node with branching factor $k$ consists of $k-1$ label
hashes, $k-1$ values, and $k$ vORAM identifiers which represent
pointers to the child nodes. This is described in
Figure~\ref{fig:hirbnode} where $h_i$ indicates $\Hash(\dskey_i)$.

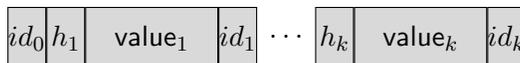
\begin{figure}[htbp]
  \begin{center}\begin{tikzpicture}[node distance=0]
  \def\inode#1#2{%
  \node[short,label=center:{#2}] (#1) {}}
  \def\shnode#1#2#3{%
  \node[short,right=of #1, label=center:{#3}] (#2) {}}
  \def\lnode#1#2{%
    \node[long,right=of #1] (#2) {}}
  \def\llnode#1#2#3{%
  \node[long,right=of #1, label=center:{#3}] (#2) {}}
  \def\mnode#1#2#3{%
  \node[medium,right=of #1, label=center:{#3}] (#2) {}}

  \inode{a}{$id_0$};
  \shnode{a}{b}{$h_1$};
  \llnode{b}{c}{$\dsval_1$};
  \shnode{c}{d}{$id_1$};
  \node[short,draw=none,fill=none,right=of d,text height=0pt,text depth=0pt,text
  width=0.5cm] (g) {$\ldots$};
  \shnode{g}{i}{$h_k$};
  \llnode{i}{j}{$\dsval_k$};
  \shnode{j}{k}{$id_k$};
  \end{tikzpicture}\end{center}
  \vspace{-1em}
  \caption{HIRB node with branching factor $k$.\label{fig:hirbnode}}
\end{figure}

Similar to the vORAM itself, the length of the hash function should be
long enough to reduce the probability of collision below 
$2^{-\hashparam}$, so define
$|\Hash(\dskey)| = \max(2H\lg\beta + \hashparam, \secparam)$,
and define $\nodesize_k$ to be the size of a HIRB tree
node with branching factor $k$, given as
$$\nodesize_k = (k+1)(2T+\hashparam+1) + k( |\Hash(\dskey)| +|\dsval|),$$
where we write $|\dsval|$ as an upper bound on the size of the largest
value stored in the HIRB. (Recall that the size of each vORAM
identifier is $2T+\hashparam+1$.)  Each HIRB tree \emph{node} will be
stored as a single \emph{block} in the vORAM, so that a HIRB node with
branching factor $k$ will ultimately be a vORAM block with length
$\nodesize_k$.

As $\beta$ reflects the expected branching factor of a node,
it must be an integer greater than or equal to 1. This parameter
controls the efficiency of the tree and should be chosen according to
the size of vORAM buckets. In particular, using the results of
Theorem~\ref{thm:oramstash} in the previous section, and the HIRB node
size defined above,
one would choose $\beta$ according to the inequality
$20\nodesize_\beta \le Z$,
where $Z$ is the size of each vORAM bucket.  According to our
experimental results in Section~\ref{sec:eval}, the constant $20$ may
be reduced to $6$.

The height $H$ must be set so that $H \ge \log_\beta n$; otherwise we
risk the root node growing too large.  We assume that $H$ is fixed at
all times,
which is easily handled when an upper bound $n$ is known \emph{a priori}. 

\paragraph{HIRB tree operations.}
As previously described, the entries in a HIRB node are sorted by the
hash of the labels, and the \emph{search path} for a label is also
according to the label hashes. A lookup operation for a label requires
fetching each HIRB node along the search path from the vORAM and
returning the matching \dsval{}. 

Initially, an empty HIRB tree of height $H$ is created, as shown in
Figure~\ref{fig:emptyhirb}. Each node has a branching factor of 1 and
contains only the single vORAM identifier of its child.

\begin{figure}[htbp]
  \begin{center}\begin{tikzpicture}[level distance=25 pt, scale=0.9]
    \tikzstyle{hirb}=[rectangle split, rectangle split horizontal,rectangle split ignore empty parts,draw, fill=white]
    \tikzstyle{every node}=[hirb]

    \node (a) {\phantom{}{$\oslash$}} [->]
      child {node {\phantom{}{$\oslash$}} [->]
        child {node[draw=none] {$\vdots$} [->]
          child {node (z) {\phantom{}{$\oslash$}}}}};
    \draw [decorate,decoration={brace,amplitude=10pt}]
      (a.north east) -- node[draw=none,right=10 pt]{\small $H+1$ nodes} 
      (z.south east);
  \end{tikzpicture}\end{center}
  \vspace{-.1in}
  \caption{Empty HIRB with height $H$.\label{fig:emptyhirb}}

\end{figure}
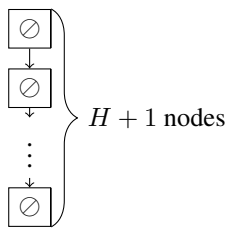

Modifying the HIRB with a \set{} or \delete{} operation on some
\dskey{} involves first computing the height of the \dskey{}. The height is
determined by sampling from a geometric distribution with probability
$(\beta-1)/\beta$, which we derandomize by using a pseudorandom sequence based
on $\Hash(\dskey)$. The distribution guarantees that, in expectation, the number
of items at height 0 (i.e., in the leaves) is $\tfrac{\beta-1}{\beta}n$, the
number of items at height 1 is $\tfrac{\beta-1}{\beta^2}n$, and so on.

Inserting or removing an element from the HIRB involves (respectively)
splitting or merging nodes along the search path from the height of
the item down to the leaf. This differs from a typical B-tree in that
rather than inserting items at the leaf level and propagating up or
down with splitting or merging, the HIRB tree requires that the heights
of all items are fixed. As a result, insertions and deletions occur at
the selected height within the tree according to the label hash. A
demonstration of this process is provided in
Figure~\ref{fig:splitmerge}.

\begin{figure}[t]
  \begin{center}\begin{tikzpicture}[scale=0.8]
    \tikzstyle{hirb}=[rectangle split, rectangle split horizontal,rectangle split ignore empty parts,draw, fill=white]
    \tikzstyle{every node}=[hirb]
    \tikzstyle{level 1}=[sibling distance=33 pt]
    \tikzstyle{level 2}=[sibling distance=53 pt]
    \tikzstyle{level 3}=[sibling distance=23 pt]

    \node (a) {\phantom{X}\nodepart{two}\phantom{X}%
               \nodepart{three}\phantom{X}} [->]
      child {node[draw=none] {$\ldots$}}
      child {node[draw=none] {$\ldots$}}
      child {node [very thick] {\phantom{X}\nodepart{two}\phantom{X}} [->]
        child {node[draw=none] {$\ldots$}}
        child {node [very thick] {\phantom{X}\nodepart{two}\phantom{X}%
                     \nodepart{three}\phantom{X}\nodepart{four}\phantom{X}}
                     [->]
          child {node[draw=none] {$\ldots$}}
          child {node [very thick] {\phantom{X}}}
          child {node[draw=none] {$\ldots$}}
          child {node[draw=none] {$\ldots$}}
          child {node[draw=none] {$\ldots$}}
        }
        child {node[draw=none] (ay) {$\ldots$}}
      }
      child {node[draw=none] (ax) {$\ldots$}}
    ;

    \tikzstyle{level 1}=[sibling distance=43 pt]
    \tikzstyle{level 2}=[sibling distance=43 pt]

    \node [right=2cm of a] (b) {\phantom{X}\nodepart{two}\phantom{X}%
                                \nodepart{three}\phantom{X}} [->]
      child {node[draw=none, right=3pt of ax] {$\ldots$}}
      child {node[draw=none] {$\ldots$}}
      child {node [very thick] {\phantom{X}\nodepart{two}$X$%
                   \nodepart{three}\phantom{X}} [->]
        child {node[draw=none, right=3pt of ay] {$\ldots$}}
        child {node [very thick] (c1) {\phantom{X}}
          child {node[draw=none] {$\ldots$}}
          child {node[very thick] {\phantom{X}}}
        }
        child {node [very thick,right=24pt of c1] (c2) {\phantom{X}\nodepart{two}\phantom{X}%
                     \nodepart{three}\phantom{X}}
                     [->]
          child {node[very thick] {\phantom{}{$\oslash$}}}
          child {node[draw=none] {$\ldots$}}
          child {node[draw=none] {$\ldots$}}
          child {node[draw=none] {$\ldots$}}
        }
        child {node[draw=none,right=6pt of c2] {$\ldots$}}
      }
      child {node[draw=none] {$\ldots$}}
    ;
  \end{tikzpicture}\end{center}
\caption{HIRB insertion/deletion of $X = (\Hash(\dskey),\allowbreak\dsval)$: On
  the left is the HIRB without item $X$, displaying only the nodes
  along the search path for $X$, and on the right is the state of the
  HIRB with $X$ inserted. Observe that the insertion operation (left
  to right) involves splitting the nodes below $X$ in the HIRB, and
  the deletion operation (right to left) involves merging the nodes
  below $X$.  }
    \label{fig:splitmerge}
\end{figure}
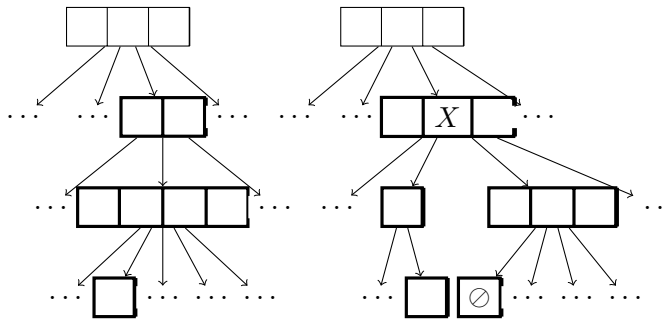

In a HIRB tree with height $H$, each \get{} operation requires reading {\em
exactly} $H+1$ nodes from the vORAM, and each \set{} or \delete{} operation
involves reading and writing at most $2H+1$ nodes. To support obliviousness,
each operation will require exactly $2H+1$, accomplished by padding with
``dummy'' accesses so that every operation has an indistinguishable access
pattern.

One way of reading and updating the nodes along the search path 
would be to read all $2H+1$ HIRB nodes from
the vORAM and store them in temporary memory and then write back the
entire path after any update.
However, properties of the HIRB tree enable better performance because the
height of each HIRB tree element is uniquely determined, which means we can
perform the updates \emph{on the way down} in the search path. This only
requires 2 HIRB tree nodes to be stored in local memory at any given time.

Facilitating this extra efficiency requires
considerable care in the implementation due to the nature of vORAM
identifiers; namely, each internal node must be written back to vORAM
before its children nodes are fetched.  Fetching children nodes will
{\em change} their vORAM identifiers and invalidate the pointers in
the parent node. The solution is to \emph{pre-generate} new random
identifiers of the child nodes before they are even accessed from the
vORAM. 

The full details of the HIRB operations can be found in Appendix~\ref{app:hirb}.

\paragraph{HIRB tree properties.}
For our analysis of the HIRB tree, we first need to understand the
distribution of items among each level in the HIRB tree. We assume a subroutine
$\chooseheight(\dskey)$ evaluates a random function on $\dskey$ to generate
random coins, using which it samples from a truncated geometric distribution
with maximum value $H$ and probability $(\beta-1)/\beta$.

\begin{restatable}{assumption}{hirblevelass}\label{ass:hirblevel}
  If ~$\dskey_1,\ldots,\dskey_n$ are any $n$ distinct labels stored in a HIRB,
  then the heights $$\chooseheight(\dskey_1),\ldots,\allowbreak
  \chooseheight(\dskey_n)$$ are independent random samples from a truncated
  geometric distribution over $\{0,1,\ldots,H\}$ with probability
  $(\beta-1)/\beta$, where the randomness is determined entirely by the
  the random oracle and the random function upon creation of the
  HIRB.  
\end{restatable}

In practice, the random coins for $\chooseheight(\dskey)$ are prepared by
computing $\sf coins$ = PRG(SHA1($seed \| \dskey$)), where $seed$ is a global
random seed, and PRG is a pseudorandom generator. With SHA1 modeled as a random
oracle, the $\sf coins$ will be pseudorandom.

\begin{theorem}\label{thm:hirb}
  The HIRB tree is a dictionary data structure that associates arbitrary
  \dskey{}s to \dsval{}s. If it contains $n$ items, and has height
  $H \ge \log_\beta n$, and the nodes are stored in a vORAM,
  then the following properties hold:
  \begin{itemize}[noitemsep,topsep=0pt]
    \item The probability of
      failure in any operation is at most $2^{-\hashparam}$.
    \item Each operation requires exactly $2H+1$ node accesses, only 2
      of which need to be stored in temporary memory at any given time.
    \item The data structure itself, \emph{not} counting the pointers, is
      strongly history independent.
  \end{itemize}
\end{theorem}
  The first property follows from the fact that the only way the HIRB
  tree can fail to work properly is if there is a hash collision.
  Based on the hash length defined above,
  the probability that any $2$ keys collide amongst the $n$ labels in
  the HIRB is at most $2^{-\hashparam}$.
  The second property follows from the description of the operations
  \get{}, \set{}, and \delete{}, and is crucial not only for the
  performance of the HIRB but also for the obliviousness property.
  The third property is a consequence of the fact that the HIRB is
  uniquely represented up to the pointer values, 
  after the hash function is chosen at initialization.

\paragraph{vORAM+HIRB properties.}
We are now ready for the main theoretical results of the paper, which have to
do with the performance and security guaranteed by the vORAM+HIRB
construction. 
\fullversion{These proofs follow in a straightforward way from the
results we have already stated on vORAM and on the HIRB, so we leave
their proofs to Appendix~\ref{app:hirb}.}

\begin{restatable}{theorem}{hirbperf}
  \label{thm:hirbperf}
  Suppose a HIRB tree with $n$ items and height $H$
  is stored within a vORAM with $L$ levels, bucket size $Z$,
  and stash size $R$. Given choices for $Z$ and $\hashparam>0$,
  set the parameters as follows:
  \begin{align*}
    T &\ge \lg(4n + \lg n + \hashparam) \\
    \beta &= \max\{\beta | Z \ge 20 \cdot \nodesize_\beta\} \\
    R &\ge \gamma \cdot \nodesize_\beta \\
    H &\ge \log_\beta n
  \end{align*}
  Then the probability of failure due to stash overflow or collisions
  after each operation is at most
  \[\Pr[\text{vORAM+HIRB failure}] \le 30\cdot (0.883)^\hashparam.\]
\end{restatable}

The parameters follow from the discussion above. Again note that the
constants $30$ and $0.883$ are technical artifacts of the analysis.

\begin{restatable}{theorem}{hirbsec}
  Suppose a vORAM+HIRB is constructed with parameters as above.  The vORAM+HIRB
  provides obliviousness, secure deletion, and history independence with leakage
  of $O(n + n\secparam/(\log n))$ operations. 
\end{restatable}

The security properties follow %
from the previous results on the vORAM and the HIRB. Note that the
HIRB structure itself provides history independence with no leakage,
but when combined with the vORAM, the pointers may leak information about recent
operations. The factor $O(\log n)$ difference from the amount of
leakage from vORAM in Theorem~\ref{thm:oramhi} arises because each
HIRB operation entails $O(\log n)$ vORAM operations. Following the
discussion after Theorem~\ref{thm:oramhi}, we could also reduce the
leakage in vORAM+HIRB to $O(n/\log^2 n)$, with constant-factor increase
in bandwidth, which again is optimal
according to Theorem~\ref{thm:oramhilb}.

\section{Evaluation}
\label{sec:eval}

We completed two empirical analyses of the vORAM+HIRB system. First,
we sought to determine the most effective size for vORAM buckets with
respect to the expected block size, i.e., the ratio $Z/B$.  Second, we
made a complete implementation of the vORAM+HIRB and measured its
performance in storing a realistic dataset of key/value pairs of 22MB in size. 
The complete source code of our implementation is available upon
request. 

\subsection{Optimizing vORAM parameters}
A crucial performance parameter in our vORAM construction is the ratio
$Z/B$ between the size $Z$ of each bucket and the expected size $B$ of
each block.  (Note that $B = \nodesize_\beta$ when storing HIRB nodes
within the vORAM.)  This ratio is a constant factor in the bandwidth
of every vORAM operation and has a considerable effect on performance.
In the Path ORAM, the best corresponding theoretical ratio is $5$,
whereas it has been shown experimentally that a ratio of $4$ will also
work, even in the worst case \cite{CCS:SDSFRY13}.

We performed a similar experimental analysis of the ratio $Z/B$ for
the vORAM. Our best theoretical ratio from Theorem~\ref{thm:oramstash}
is 20, but as in related work, the experimental performance is
better. The goal is then to find the optimal, empirical choice for the
ratio $Z/B$: If $Z/B$ is too large, this will increase the overall
communication cost of the vORAM, and if it is too small, there is a
risk of stash overflow and loss of data or obliviousness.

For the experiments described below, we implemented a vORAM structure without
encryption and inserted a chosen number of variable size blocks whose
sizes were randomly sampled from a geometric distribution with expected
size 68 bytes.
To avoid collisions, we ensured the identifier lengths
satisfied $\hashparam\ge40$.

\paragraph{Stash size.} 
To analyze the stash size for different $Z/B$ ratios, we ran a number of
experiments and monitored the maximum stash size observed at any point
throughout the experiment.  Recall, while the stash will typically be empty
after every operation, the max stash size should grow logarithmically with
respect to the number of items inserted in the vORAM. The primary results are
presented in Figure~\ref{fig:maxstash}.

\graphmaxstash
This experiment was conducted by running 50 simulations of a vORAM with $n$
insertions and a height of $T=\lg n$. The $Z/B$ value ranged from 1 to 50, and
results in the range 1 through 12 are presented in the graph for values of $n$
ranging from $10^2$ through $10^5$. 
The graph plots the ratio $R/\lg n$, where $R$ is the
largest max stash size at any point in any of the 50 simulations.
Observe that
between $Z/B=4$ and $Z/B=6$ the ratio stabilizes for all values of $n$,
indicating a maximum stash of approximately $100 \lg n$.

\graphmttf
In order to measure how much stash would be needed in practice for much
larger experimental runs, we fixed $Z/B=6$ and for three large database
sizes, $n=2^{18}, 2^{19}, 2^{20}$, For each size, we executed $2n$
operations, measuring the size of stash after each. In practice, as we
would assume from the theoretical results, the stash size is almost
always zero. However, the stash does occasionally become non-empty, and
it is precisely the frequency and size of these rare events that we wish
to measure.

\cref{fig:mttf} shows the result of our stash overflow experiment. We
divided each test run of roughly $2n$ operations into roughly $n$
overlapping windows of $n$ operations each, and then for each window,
and each possible stash size,
calculated the number of operations before the first time that stash
size was exceeded. The average number of operations until this occurred,
over all $n$ windows, is plotted in the graph. The data shows a linear
trend in log-scale, meaning that the stash size necessary to ensure low
overflow probability after $N$ operations is $O(\log N)$, as
expected. Furthermore, in all experiments we never witnessed a stash
size larger than roughly 10KB, whereas the theoretical bound of 
$100\lg n$ items would be 16KB for the largest test with $2^{20}$
8-byte items.

\paragraph{Bucket utilization.}
Stash size is the most important parameter of vORAM, but
it provides a limited view into the optimal bucket size ratio, in particular as
the stash overflow is typically zero after every operation,
for sufficiently large buckets.  We
measured the utilization of buckets at different levels of the vORAM with varied
heights and $Z/B$ values. The results are presented in Figure~\ref{fig:loadutil}
and were collected by averaging the final bucket utilization from 10
simulations. The utilization at each level is measured by dividing the total
storage capacity of the level by the number of bytes at the level. In all cases,
$n=2^{15}$ elements were inserted, and the vORAM height varied between 14, 15, and
16. The graph shows that with height $\lg n = 15$ 
or higher and $Z/B$ is 6 or higher,
utilization stabilizes throughout all the levels (with only a small  spike at the
leaf level). 

\graphloadutil

The results indicate, again, that when $Z/B=6$, the
utilization at the interior buckets stabilizes.  With smaller ratios, e.g.,
$Z/B=4$, the utilization of buckets higher in the tree dominates those lower in
the tree; essentially, blocks are not able to reach lower levels resulting in
higher stash sizes (see previous experiment). With larger ratios,
which we measured all the way to $Z/B=13$, we observed consistent
stabilization.

In addition, our data shows that decreasing the number of levels from
$\lg n$ to $\lg n - 1$ (e.g., from 15 to 14 in the figure) increases
utilization at the leaf nodes as expected (as depicted in the spike in
the tail of the graphs), but when $Z/B\ge 6$ the extra blocks in leaf
nodes do \emph{not} propagate up the tree and affect the stash. It
therefore appears that in practice, the number of levels $T$ could be
set to $\lg n - 1$, which will result in a factor of 2 savings in the
size of persistent (cloud) storage due to high utilization at the leaf
nodes. This follows a similar observation about the height of the Path 
ORAM made by~\cite{CCS:SDSFRY13}.

\subsection{Measuring vORAM+HIRB Performance}

We measured the performance of our vORAM+HIRB implementation on a real
data set of reasonable size, and compared to some alternative methods
for storing a remote map data structure that provide varying levels of
security and efficiency. All of our implementations used the same
client/server setup, with a Python3 implementation and AWS as the cloud
provider, in order to give a fair comparison.

\paragraph{Sample dataset.} 
We tested the
performance of our implementation on a dataset of 300,000 synthetic
key/value pairs where where keys were variable sizes (in the order of
bytes 10-20 bytes) and values were fixed at 16 bytes. The total
unencrypted data set is 22MB in size. In our experiments, we used some
subset of this data dependent on the size of the ORAM, and for each
size, we also assumed that the ORAM user would want to allow the
database to grow. As such, we built the ORAM to double the size of the
initialization.

\paragraph{Optimized vORAM+HIRB implementation.}
We fully implemented our vORAM+HIRB map data structure using Python3 and Amazon
Web Services as the cloud service provider. We used AES256 for encryption in
vORAM, and used SHA1 to generate labels for the HIRB. In our setting, we
considered a client running on the local machine that maintained the erasable
memory, and the server (the cloud) provided the persistent storage with a simple
get/set interface to store or retrieve a given (encrypted) vORAM bucket.

For the vORAM buckets, we choose $Z/B=6$ based on the prior
experiments, and a bucket size of 4K, which is the preferred back-end
transfer size for AWS, and was also the bucket size used by
\cite{CCS:BNP+15}. One of the advantages of the vORAM over other
ORAMs is that the bucket size can be set to match the storage
requirements with high bucket utilization. The settings for the HIRB
were then selected based on Theorem~\ref{thm:hirbperf} and based on
that, we calculated a $\beta=12$ for the sample data (labels and
values) stored within the HIRB. The label, value, and associated vORAM
identifiers total 56 bytes per item.

In our experiments, we found that 
the round complexity of protocols dominate performance and so we made a number
of improvements and optimizations to the vORAM access routines to compensate.
The result is an optimized version of the vORAM.  In particular:

\begin{itemize}
\item {\it Parallelization}: The optimized vORAM transfers buckets along a
        single path in parallel over simultaneous connections for both the
        \evict{} and \writeback{} methods. Our experiments used up to $T$
        threads in parallel to fetch and send ORAM block files, and each
        maintained a persistent {\tt sftp} connection. 

\item {\it Buffering}: A local buffer storing $2T$ top-most ORAM buckets  was used to facilitate
    asynchronous path reading and writing by our threads. Note the size of the
    client storage still remains with $O(\log n)$ since $T = O(\log n)$. This
    had an added performance benefit beyond the parallelization because the top
    few levels of the ORAM generally resided in the buffer and did not need to
    be transferred back and forth to the cloud after every operation.  
\end{itemize}

These optimizations had a considerable effect on the performance. We did not include the cost of the $\approx$ 2 second
setup/teardown time for these SSL connections in our results as these
were a one-time cost incurred at initialization. 
Many similar techniques to these
have been used in previous work to achieve similar performance gains
(e.g., \cite{lorch2013shroud,YRFKDD14superblock}), although they have not been
previously applied to oblivious data structures.

\paragraph{Comparison baselines.}
We compared our optimized vORAM+HIRB construction with four other
alternative implementations of a remote map data structure,
with a wide range of performance and security properties:

\begin{itemize}
\item \emph{Un-optimized vORAM+HIRB}. This is the same as our normal
vORAM+HIRB construction, but without any buffering of vORAM buckets and
with only a single concurrent sftp connection. This comparison allows us
to see what gains are due to the algorithmic improvements in vORAM and
HIRB, and which are due to the network optimizations.

\item {\it Naive Baseline}: We implemented a naive approach
  that provides all three security properties, obliviousness, secure
  deletion, and history independence. The method involves maintaining
  a single, fixed-size encrypted file transferred back and forth
  between the server and client and re-encrypted on each access. 
  While this solution is cumbersome for large sizes, it is the obvious
  solution for small databases and thus provides a useful baseline.
  Furthermore, we are not aware of any other method
  (other than vORAM+HIRB) to provide obliviousness, secure-deletion,
  and history independence.

\item {\it ORAM+AVL}: We implemented the ODS proposed by
  \cite{CCS:WNLCSS14} of
  an AVL embedded within an non-recursive Path ORAM. Note that
  ORAM+AVL does not provide secure deletion nor history
  independence. We used the same cryptographic settings as
  our vORAM+HIRB implementation, and used 256 byte blocks for each AVL
  node, which was the smallest size we could achieve without additional
  optimizations. As
  recommended by~\cite{CCS:SDSFRY13}, we stored $Z=4$ fixed-size blocks
  in each bucket, for a total of 1K bucket size.
  Note that this bucket size is less than the 4K transfer size
  recommended by the cloud storage, which reflects the limitation of
  ORAM+AVL in that it cannot effectively utilize larger buckets.
  We add the observation that, when the same experiments were run with
  4K size buckets (more wasted bandwidth, but matching the other experiments), 
  the timings did
  not change by more than 1 second, indicating that the 4K bucket size
  is a good choice for the AWS back-end.

\item {\it SD-B-Tree}: As another comparison point, we implemented a
  remotely stored B-Tree with secure deletion where each node is
  encrypted with a key stored in the parent with re-keying for each
  access, much as described by Reardon et al.~\cite{CCS:RRBC13}. 
  While this solution provides secure deletion, and stores all data
  encrypted, it does not provide obliviousness nor history independence.
  Again, we used AES256 encryption,
  with $\beta=110$ for the B-tree max internal node size in order to
  optimize 4K-size blocks.
\end{itemize}

In terms of security, only our vORAM+HIRB as well as the naive baseline
provide obliviousness, secure deletion, and history independence. The
ORAM+AVL provides obliviousness only, and the SD-B-Tree is most
vulnerable to leaking information in the cloud, as it provides secure
deletion only. 

In terms of asymptotic performance, the SD-B-Tree
is fastest, requiring only $O(\log n)$ data transfer per operation. The
vORAM+HIRB and ORAM+AVL both require $O(\log^2 n)$ data transfer per
operation, although as discussed previously the vORAM+HIRB saves a
considerable constant factor. The naive baseline requires $O(n)$
transfer per operation, albeit with the smallest possible constant factor.

\graphaccesstime

\paragraph{Experimental  results.}
The primary result of the experiment is presented in
Figure~\ref{fig:accesstime} where we compared the cost of a single
access time (in seconds) across the back end storage (note, graph is
log-log). Unsurprisingly, the SD-B-Tree implementation is fastest for
sufficiently large database sizes. However, our optimized vORAM+HIRB
implementation was competitive to the SD-B-Tree performance, both being
less than 1 second across our range of experiments.

Most striking is the access time of ORAM+AVL compared to the vORAM+HIRB
implementations. In both the optimized and un-optimized setting, the
vORAM+HIRB is orders of magnitude faster than ORAM+AVL, 20X faster
un-optimized and 100X faster when optimized. Even for a relatively small
number of entries such as $2^{11}$, a single access of ORAM+AVL takes
35 seconds, while it only requires 1.3 seconds of un-optimized
vORAM+HIRB and 0.2 second of an optimized implementation. It is not
until $2^{19}$ entries that ORAM+AVL even outperforms the naive $O(n)$
baseline solution.

As described previously, we attribute much of the speed to decreasing
the round complexity. The HIRB tree requires much smaller
height as compared to an AVL tree because each HIRB node contains
$\beta$ items on average as compared to just a single item for an
AVL tree. Additionally, the HIRB's height is fixed and does not require
padding to achieve obliviousness. Each AVL operation entails
$3\cdot1.44\lg N$ ORAM operations as compared to just
$2\log_\beta N$ vORAM operations for the HIRB.  This difference
in communication cost is easily observed in Table~\ref{tab:comm}.
Overall, we see that the storage and communication costs for vORAM+HIRB
are not too much larger than that for a secure deletion B-tree, which
does not provide any access pattern hiding as the oblivious
alternatives do.

\tablecomm

(The values in this table were generated by considering the worst-case
costs in all cases, for our actual implementations, but considering
only a single operation. Note that, for constructions providing
obliviousness, every operation must \emph{actually} follow this worst
case cost, and so the comparison is fair.)

Put simply, the vORAM+HIRB and SD-B-Tree are the only implementations
which can be considered practical for real data sizes, and the benefit
of vORAM+HIRB is its considerable additional security guarantees of
oblivious and bounded history independence.

\section{Conclusion}
\label{sec:conclusion}

In this paper, we have shown a new secure cloud storage system combining the
previously disjoint security properties of obliviousness, secure deletion, and
history independence. This was accomplished by developing a new variable block
size ORAM, or vORAM, and a new history independent, randomized data structure
(HIRB) to be stored within the vORAM.

The theoretical performance of our vORAM+HIRB construction is competitive to
existing systems which provide fewer security properties. Our implemented
system is up to 100X faster (w.r.t. access time) than current best oblivious
map data structure (which provides no secure deletion or history independence)
by Wang et al.~(CCS 14), bringing our single-operation time for a
reasonable-sized database ($>2^{19}$) to less than 1 second per access.

There much potential for future work in this area. For example, one
could consider data structures that support a richer set of
operations, such as range queries, while preserving obliviousness,
secure deletion, and history independence.  Additionally, the vORAM
construction in itself may provide novel and exciting new analytic
results for ORAMs generally by not requiring fixed bucket sizes. There
is a potential to improve the overall utilization and communication
cost compared to existing ORAM models that used fixed size blocks.

Finally, while we have demonstrated the practicality in terms of
overall per-operation speed, we did not consider some additional
practical performance measures as investigated by \cite{CCS:BNP+15},
such as performing asynchronous operations and optimizing upload vs
download rates.  Developing an ODS map considering these concerns as
well would be a useful direction for future work.

\section*{Acknowledgements}
\addcontentsline{toc}{section}{Acknowledgements}
The authors are supported by the Office of Naval Research (ONR),
as well as National Science Foundation awards
\#1406192 and \#1319994.

\renewcommand{\bibpreamble}{\addcontentsline{toc}{section}{References}}

{
\bibliographystyle{plainnat}


}

\begin{appendices}
\section{vORAM Operation Details}
\label{app:voram}

\ifpreprint  %
  The full detail of the vORAM helper functions is provided in
  \cref{fig:voram1}, and the three main operations are shown in
  \cref{fig:voram2}.

  \begin{figure}[p]
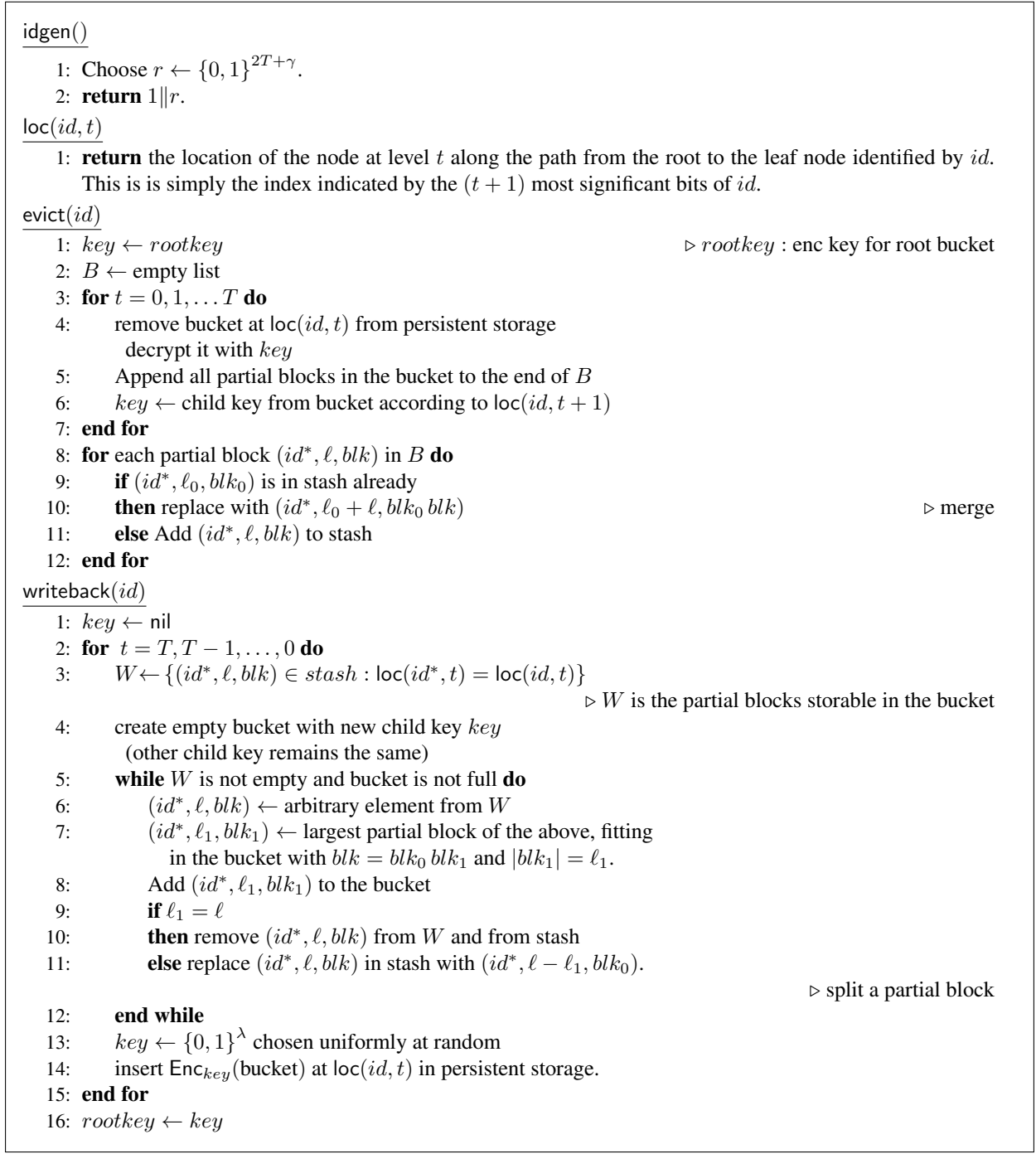

  \begin{minipage}{\textwidth}
  \centering
  \begin{framed}
  \begin{minipage}{6 in}
\else  %
  We give the pseudocode of vORAM helper functions.

  \begingroup
\fi %

\newcommand{\espace}{\hspace{2em}}

\small
\hspace{-1em}$\underline{\idgen()}$ 

\begin{algorithmic}[1]
\State Choose $r \gets \zo^{2T+\hashparam}$.
\State \Return $1 \| r$. 
\end{algorithmic}

\hspace{-1em}$\underline{\path(id, t)}$ 

\begin{algorithmic}[1]
\State \Return the location of the node at level $t$ along the path from
the root to the leaf node identified by $id$. This is is simply the
index indicated by the $(t+1)$ most significant bits of $id$.
\end{algorithmic}

\hspace{-1em}$\underline{\evict(id)}$  

\begin{algorithmic}[1]
\State $key \gets rootkey$ \Comment $rootkey:$ enc key for root bucket 
\State $B \gets$  empty list
\For {$t = 0, 1, \ldots T$}
    \State remove bucket at $\path(id, t)$ from persistent storage
    \Statex \espace decrypt it with $key$
  \State Append all partial blocks in the bucket to the end of $B$
  \State $key \gets $ child key from bucket according to $\path(id,t+1)$
\EndFor

\For {each partial block $(id^*, \ell, blk)$ in $B$}
  \State \textbf{if} {$(id^*, \ell_0, blk_0)$ is in stash already}
  \State \textbf{then} replace with $(id^*, \ell_0+\ell, blk_0\, blk)$ \Comment{merge} 
  \State \textbf{else} Add $(id^*, \ell, blk)$ to stash
\EndFor
\end{algorithmic}

\hspace{-1em}$\underline{\writeback(id)}$  

\begin{algorithmic}[1]
\State $key \gets \nil$
\For { $t = T, T-1, \dots, 0$} 
    \State $W \from \{ (id^*, \ell, blk) \in \stash: \path(id^*, t) =
    \path(id,t) \}$ 
    \Statex \Comment $W$ is the partial blocks storable in the bucket
    \State create empty bucket with new child key $key$ 
    \Statex \espace (other child key remains the same)
    \While {$W$ is not empty and bucket is not full}
      \State $(id^*, \ell, blk) \gets $ arbitrary element from $W$
      \State $(id^*, \ell_1, blk_1) \gets$ largest partial block of the above,
      fitting \Statex \espace \espace in the bucket with $blk=blk_0\,blk_1$ and
      $|blk_1|=\ell_1$.
      \State Add $(id^*, \ell_1, blk_1)$ to the bucket
      \State \textbf{if} {$\ell_1 = \ell$} 
      \State \textbf{then} remove $(id^*,\ell,blk)$ from $W$ and from stash
      \State \textbf{else}
          replace $(id^*,\ell,blk)$ in stash with
          $(id^*,\ell-\ell_1, blk_0)$. 
      \Statex \Comment{split a partial block}
    \EndWhile
    \State $key \gets \zo^\secparam$ chosen uniformly at random
    \State insert $\Enc_{key}(\text{bucket})$ at $\path(id,t)$ in
    persistent storage.
\EndFor
\State $rootkey \gets key$
\end{algorithmic}

\ifpreprint  %
  \end{minipage}
  \end{framed}
  \end{minipage}
  \caption{vORAM helper functions\label{fig:voram1}}
  \end{figure}
\else  %
  \endgroup
\fi %

\ifpreprint  %
  \begin{figure}[tbp]
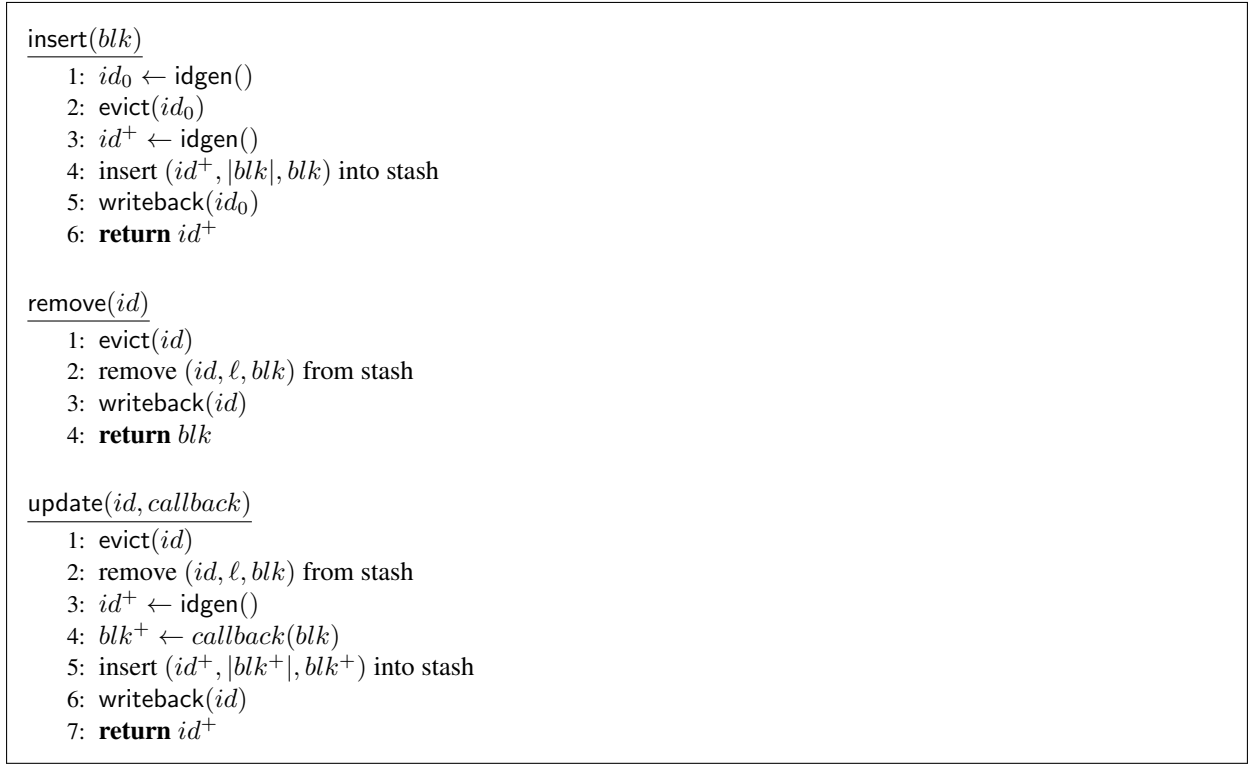

  \begin{minipage}{\textwidth}
  \centering
  \begin{framed}
  \begin{minipage}{6 in}
\else  %
  Now, we give the pseudocode of the vORAM operations: 

  \begingroup
\fi %

\small
\hspace{-1em}$\underline{\ins(blk)}$

\begin{algorithmic}[1]
\State $id_0 \gets \idgen()$
\State $\evict(id_0)$
\State $id^+ \gets \idgen()$
\State insert $(id^+, |blk|, blk)$ into stash
\State $\writeback(id_0)$
\State \Return $id^+$
\end{algorithmic}
\bigskip

\hspace{-1em}$\underline{\remove(id)}$

\begin{algorithmic}[1]
\State $\evict(id)$
\State remove $(id, \ell, blk)$ from stash
\State $\writeback(id)$
\State \Return $blk$
\end{algorithmic}
\bigskip

\hspace{-1em}$\underline{\upd(id, callback)}$

\begin{algorithmic}[1]
\State $\evict(id)$\
\State remove $(id, \ell, blk)$ from stash
\State $id^+ \gets \idgen()$
\State $blk^+ \gets callback(blk)$
\State insert $(id^+, |blk^+|, blk^+)$ into stash
\State $\writeback(id)$
\State \Return $id^+$
\end{algorithmic}

\ifpreprint  %
  \end{minipage}
  \end{framed}
  \end{minipage}
  \caption{vORAM operations\label{fig:voram2}}
  \end{figure}
\else  %
  \endgroup
\fi %

\section{HIRB Operation Details}
\label{app:hirb}

\ifpreprint %
We described the HIRB data structure in \cref{sec:ds}.  The full details
of the different subroutines are provided in 
\cref{fig:hirbpath,fig:hirbops}.
\fi %

All the HIRB tree
operations depend on a subroutine \hirbpath{}, 
which given a label hash, HIRB root node identifier, and vORAM,
generates tuples $(\ell,v_0,v_1,cid_1^+)$
corresponding to the search path for that label in the HIRB.
In each tuple, $\ell$ is the level of node $v_0$, which is
along the search path for the label. In the $\it initial$ part of the search
path, that is, before the given label hash is found, node $v_1$ is always
\nil{}, a dummy access used to preserve obliviousness.  The value $cid_1^+$ is the
\emph{pre-generated} identifier of the new node that will be inserted on the
\emph{next level}, for possible inclusion in one of the parent nodes as a child
pointer. This pre-generation is important, as discussed in Section~\ref{sec:ds},
so that only 2 nodes need to be stored in local memory at any given time.

When the given label hash is found, the search path
splits into two below that node, and nodes $v_0$ and $v_1$ will be the nodes on
either side of that hash label.  Note that in the actual
implementation of \hirbpath{}, $v_0$ (resp. $v_1$, if defined) corresponds to a
vORAM block, evicted with identifier $id_0$ (resp. $id_1$) and taken out from
vORAM stash.
When each tuple $(\ell,v_0,v_1,cid_1^+)$ is returned from the generator,
the two nodes can be modified by the calling function, and the modified nodes
will be written back to the HIRB.  If $v_1$ is returned from \hirbpath{} as
\nil{}, but is then modified to be a normal HIRB node, that new node is
subsequently inserted into the HIRB. 
\def\vORAM{{\sf M}}

\ifpreprint  %
  \begin{figure}[htp]
  \begin{minipage}{\textwidth}
  \centering
  \begin{framed}
  \begin{minipage}{6 in}
\fi %
{\small
    \hspace{-1em}$\underline{\hirbpath(h, rootid, \vORAM)}$ 
    \Comment{\vORAM ~is vORAM} 

\begin{algorithmic}[1]
\State $(id_0\tupsep id_0^+) \gets (rootid\tupsep \vORAM.\idgen())$
\State $rootid \gets id_0^+$
\State $(id_1\tupsep id_1^+) \gets (\vORAM.\idgen()\tupsep \vORAM.\idgen())$
  \Comment{dummy access}
\State $found \gets \textsf{false}$
\For {$\ell = 0, 1, 2, \ldots, H$}
  \State $\vORAM.\evict(id_0)$
  \State $\vORAM.\evict(id_1)$
  \State \textbf{if} {$\ell = H$} \textbf{then} 
    $(cid_0^+\tupsep cid_1^+) \gets (\nil\tupsep \nil)$
  \State \textbf{else} $(cid_0^+\tupsep cid_1^+) \gets (\vORAM.\idgen()\tupsep \vORAM.\idgen())$
  \State remove $(id_0, |v_0|, v_0)$ from $\vORAM.$stash 
  \If {$found = \textsf{true}$} 
    \State remove $(id_1, |v_1|, v_1)$ from $\vORAM.$stash
    \State $(cid_0\tupsep v_0.\child_{last}) \gets
      (v_0.\child_{last}\tupsep cid_0^+)$
    \State $(cid_1\tupsep v_1.\child_0) \gets
      (v_1.\child_0\tupsep cid_1^+)$
      \Statex \Comment $v_1$ is right next to $v_0$ at level $\ell$
  \Else
    \State $v_1 \gets \nil$ \Comment{only fetched after the target is found.}
    \State $i \gets$ index of $h$ in $v_0$ 
        \Comment $v_0.h_{i-1} < h \le v_0.h_{i}$
    \State $(cid_0\tupsep v_0.\child_i) \gets 
      (v_0.\child_i\tupsep cid_0^+)$
    \If {$v_0.h_i = h$}
      \State $found \gets \textsf{true}$
      \State $(cid_1\tupsep v_0.\child_{i+1}) \gets 
        (v_0.\child_{i+1}\tupsep cid_1^+)$
      \Statex \Comment split path: $cid_0 = v_0.\child_{i}$, $cid_1 = v_0.\child_{i+1}$ 
    \Else
      \State $cid_1 \gets \vORAM.\idgen()$
        \Comment dummy access until found
    \EndIf
  \EndIf
  \State \textbf{yield} $(\ell\tupsep v_0\tupsep v_1, cid_1^+)$
  \Statex \Comment Return to the caller, who may modify nodes.
  \State insert $(id_0^+, |v_0|, v_0)$ into $\vORAM.$stash
  \State \textbf{if} $v_1 \ne \nil$ \textbf{then}
    insert $(id_1^+, |v_1|, v_1)$ into $\vORAM.$stash
  \State $\vORAM.\writeback(id_0)$
  \State $\vORAM.\writeback(id_1)$
  \State $(id_0\tupsep id_0^+) \gets (cid_0\tupsep cid_0^+)$
  \State $(id_1\tupsep id_1^+) \gets (cid_1\tupsep cid_1^+)$
\EndFor
\end{algorithmic}
}
\ifpreprint  %
  \end{minipage}
  \end{framed}
  \end{minipage}
  \caption{Fetching the nodes along a search path in the HIRB\label{fig:hirbpath}}
  \end{figure}
\fi %

The \upd{} operation simply looks in each returned
$v_0$ along the search path for the existence of the indicated label
hash, and if found, the corresponding data value is passed to the
$callback$ function, possibly modifying it.

As with \upd{}, the \ins{} operation uses subroutine \hirbpath{}
as a generator to traverse the HIRB tree. Inserting an element from the HIRB
involves splitting nodes along the search path from the height of
the item down to the leaf. That is, for each tuple $(\ell, v_0, v_1)$ with
$\ell > \ell_h$, where $\ell_h$ is the height of the label hash $h$, if $v_1$
is $\nil$, then a new node $v_1$ is created, and the items in $v_0$ with a
label greater than $h$ are moved to a new node $v_1$.  

The \remove{} operation works similarly, but instead of splitting each
$v_0$ below the level of the found item, the values in $v_0$ and
$v_1$ are merged into $v_0$, and $v_1$ is removed by setting it
to \nil{}.

\ifpreprint  %
  \begin{figure}[htp]
  \begin{minipage}{\textwidth}
  \centering
  \begin{framed}
  \begin{minipage}{6 in}
\fi %
{\small 
\hspace{-1em}$\underline{\hirbinit(H, \vORAM)}$ 

\begin{algorithmic}[1]
\State $rootid \gets \nil$
\State $salt \gets \zo^\secparam$. Initialize $\Hash$ with $salt$. 
\For {$\ell = H, H-1, \ldots, 0$}
  \State $node \gets$ new 1-ary HIRB node with child id $rootid$
  \State $rootid \gets \vORAM.\ins(node)$
\EndFor
\State \Return $rootid$
\end{algorithmic}

\hspace{-1em}$\underline{\chooseheight(\dskey)}$ 

\begin{algorithmic}[1]
\State $h \gets \Hash(\dskey)$
\State Choose coins $(c_0, c_1, \ldots, c_{H-1}) \in \{0,1,\ldots,\beta-1\}^H$
by evaluating PRG$(h)$. 
\State \Return The largest integer $\ell\in\{0,1,\ldots,H\}$ such that
  $c_1 = c_2 = \cdots = c_\ell = 0$.
\end{algorithmic}

\medskip

\hspace{-1em}$\underline{\ins(\dskey, \dsval, rootid, \vORAM)}$  

\begin{algorithmic}[1]
\State $(h\tupsep \ell_h) \gets (\Hash(\dskey)\tupsep \chooseheight(\dskey))$
\For {$(\ell, v_0, v_1, cid_1^+) \in \hirbpath(h, rootid, \vORAM)$}
  \State $i \gets$ index of $h$ in $v_0$ \Comment $v_0.h_{i-1} < h \le v_0.h_{i}$
  \If {$v_0.h_i = h$}
    \State $v_0.\dsval_i \gets \dsval$
  \ElsIf {$\ell = \ell_h$}
  \State Insert $(h, \dsval{}, cid_1^+)$ before
    $(v_0.h_i, v_0.\dsval_i, v_0.\child_i)$
    \Statex \Comment Other items in $v_0$ are shifted over
  \ElsIf {$\ell > \ell_h$ and $v_1 = \nil$}
    \State $v_1 \gets$ new node with $v_1.\child_0 \gets cid_1^+$
    \State Move items in $v_0$ past index $i$ into $v_1$
  \EndIf
\EndFor
\end{algorithmic}

\medskip

\hspace{-1em}$\underline{\remove(\dskey, rootid, \vORAM)}$  

\begin{algorithmic}[1]
\State $(h\tupsep \ell_h) \gets (\Hash(\dskey)\tupsep \chooseheight(\dskey))$
\For {$(\ell, v_0, v_1, cid_1^+) \in \hirbpath(\dskey, rootid, \vORAM)$}
  \If {$h \in v_0$}
    \State Remove $h$ and its associated value and subtree from $v_0$
  \ElsIf {$\ell > \ell_h$ and $v_1 \ne \nil$}
    \State Add all items in $v_1$ except $v_1.\child_0$ to $v_0$
    \State $v_1 \gets \nil$
  \EndIf
\EndFor
\end{algorithmic}

\medskip

\hspace{-1em}$\underline{\upd(\dskey, callback, rootid, \vORAM)}$  

\begin{algorithmic}[1]
\State $h \gets \Hash(\dskey)$
\For {$(\ell, v_0, v_1) \in \hirbpath(h, rootid, \vORAM)$}
  \State $i \gets$ index of $h$ in $v_0$ 
  \State \textbf{if} {$v_0.h_i = h$} \textbf{then}
    $v_0.\dsval_i \gets callback(v_0.\dsval_i)$
\EndFor
\end{algorithmic}

}
\ifpreprint  %
  \end{minipage}
  \end{framed}
  \end{minipage}
  \caption{Description of HIRB tree operations.\label{fig:hirbops}}
  \end{figure}
\fi %

\section{Proofs of Important Theorems}
\label{app:proofs}
Complete proofs of our main theorems are given here.

\subsection{Proof of Theorem \ref{thm:oramhilb}} 
\label{sec:oramhi}

  Let \DDD{} be any system that stores blocks of data in persistent
  storage and erasable memory and supports \ins{} and \remove{}
  operations, accessing at most $k$ bytes in persistent or local storage
  in each \ins{} or \remove{} operation.
 
  Let $n \ge 36$ and $k \le \sqrt{n}/2$. For any $\ell \le n/(4k)$, we describe
  a PPT adversary $\AAA=(\AAA_1,\AAA_2)$ that breaks history independence with
  leakage of $\ell$ operations.  
  
  Supposing all operations are insertions, \DDD{} must access the location where
  that item's data is actually to be stored during execution of the insert
  operation, which is required to correctly store the data somehow.
  However, it may access some other locations as well to ``hide'' the access
  pattern from a potential attacker. This hiding is limited of course by $k$,
  which we will now exploit.

  The ``chooser'', $\AAA_1$, randomly chooses $n$ items which will be
  inserted; these could simply be random bit strings of equal length. 
  Call these items (and their arbitrary order) $a_1,a_2,\ldots,a_n$.
  The chooser also randomly picks an index
  $j \in \{1,2,\ldots,n-\ell-1\}$ from the beginning of the sequence.
  The operation sequence $\Vec{\op}^{(0)}$ returned by $\AAA_1$ consists
  of $n$ insertion operations for $a_1,\ldots,a_n$ in order:
  \[a_1,\ldots,a_{j-1},a_j,a_{j+1},\ldots, a_{n-\ell-1}, a_{n-\ell},
  a_{n-\ell+1},\ldots, a_n,\]
  whereas the second operation sequence $\Vec{\op}^{(1)}$ returned by
  $\AAA_1$ contains the same $n$ insertions, with only the order of 
  the $j$'th and $(n-\ell)$'th insertions swapped:
  \[a_1,\ldots,a_{j-1},a_{n-\ell},a_{j+1},\ldots, a_{n-\ell-1}, a_j,
  a_{n-\ell+1},\ldots, a_n.\]
  The adversary $\AAA_1$ includes the complete list of $a_1$ up to
  $a_n$, along with the distinguished index $j$, in the \st{} which is
  passed to $\AAA_2$.
  As the last $\ell$ operations are identical (insertion of items
  $a_{n-\ell+1}$ up to $a_n$),
  $\AAA_1$ is $\ell$-admissible.

  The ``guesser``, $\AAA_2$, looks back in the last $(\ell+1)k$
  entries in the access pattern history of persistent storage
  $\Vec{\acc}$, and tries to opportunistically decrypt the data in each 
  access entry
  using the keys from $\DDD.\mem$ (and, recursively, any other decryption
  keys which are found from decrypting data in the access pattern
  history). Some of the data may be unrecoverable, but at least the
  $\ell+1$ items which were inserted in the last $\ell+1$ operations
  \emph{must} be present in the decryptions, since their data must be
  recoverable using the erasable memory. Then the guesser simply looks
  to see whether $a_j$ is present in the decryptions; if $a_j$ is
  present then $\AAA_2$ returns 1, otherwise if $a_j$ is not present
  then $\AAA_2$ returns 0.

  In the experiment $\exphio$, $a_j$ must be among the decrypted values
  in the last $(\ell+1)k$ access entries, since $a_j$ was inserted within
  the last $\ell+1$ operations and each operation is allowed to trigger
  at most $k$ operations on the persistent storage. 
  Therefore $\Pr[\exphio=1] = 1$.

  In the experiment $\exphiz$, we know that each item
  $a_{n-\ell},\ldots,a_n$ must be present in the decryptions, and 
  there can be at most $(\ell+1)(k-1)$ other items in the decryptions.
  Since the index $j$ was chosen randomly from among the first
  $n-\ell-1$ items in the list, the probability that $a_j$ is among the
  decrypted items in this case is at most
  $$\frac{(\ell+1)(k-1)}{n-\ell-1}.$$
  From the restriction that  $\ell \le n/(4k)$, and $k \le \sqrt {n}/2 \le
  n/12$, we have $$(\ell+1)(k-1) < (\ell + 1) k
  = \ell k + k \le \tfrac{n}{4} + \tfrac{n}{12} = \tfrac{n}{3}.$$ 
  In addition, we have $n -\ell - 1 > n/2$, so the probability that $a_j$ is
  among the decrypted items is at most $\tfrac{2}{3}$, and we have
  $\Pr[\exphiz=1] \le 2/3$, and therefore $\advhi \ge 1/3$.
  According to the definition, this means that \DDD{} 
  does not provide history independence with leakage of $\ell$
  operations.

\subsection{Proof of Theorem~\ref{thm:oramstash}}
Our proofs on the distribution of block sizes in the ORAM and on the
number of HIRB nodes depend on the
following bound on the sum of geometric random variables. This is a
standard type of result along the lines of Lemma~6 in
\cite{CCS:WNLCSS14}.

\begin{lemma}\label{lem:sumgeom}
  Let $X=\sum_{1 \le i \le n}X_i$ be the sum of $n\ge 1$ independent random
  variables $X_i$,
  each stochastically dominated by a geometric distribution over
  $\{0,1,2,\ldots\}$ with expected value $\E[X_i] \le \mu$.
  Then there exists a constant $\gconst>0$ whose value depends
  only on $\mu$ such that, for any $a\ge 2$ and $b \ge 0$, we have
  \[\Pr[X \ge (\mu+1)(an + b)] < \exp(-\gconst (an+b)).\]
\end{lemma}
\begin{proof}
  By linearity of expectation, $\E[X] = \sum_{i\in[n]}\E[X_i] \le n\mu$.

  Recall that a geometric random variable with
  expected value $\mu$ is equivalent to the number of independent
  Bernoulli trials, each with probability $p=1/(\mu+1)$, before the first
  success. If $X \ge (\mu+1)(an + b)$, this is equivalent to having fewer
  than $n$ successes over 
  $k = (\mu+1)(an+b)$ independent Bernoulli trials with
  probability $p$.

  Using this formulation, we can apply the Hoeffding inequality to obtain
  \[\Pr[X \ge k] = 
    \Pr[\mathsf{Binomial}(k,p) \le n-1] < \exp(-2\epsilon^2 k),
  \]
  where $\epsilon$ is defined such that $n-1 = (p-\epsilon)k$; namely
  \[\epsilon = p-\tfrac{n-1}{k} = \tfrac{1}{\mu+1} - \tfrac{n-1}{k}.\]

  We do some manipulation:
  \begin{align*}
    2\epsilon^2 k
      &= \tfrac{2k}{(\mu+1)^2} \cdot 
        \left(1 - \tfrac{(n-1)(\mu+1)}{k}\right)^2 \\
      &= \tfrac{2(an+b)}{\mu+1} \cdot
        \left(1 - \tfrac{n-1}{an+b}\right)^2.
  \end{align*}
  Because $a\ge 2$ and $b\ge 0$, we have
  \[\tfrac{n-1}{an+b} < \tfrac{n}{an} \le \tfrac{1}{2},\]
  and so
  \[\exp(-2\epsilon^2 k) 
    < \exp\left( -\tfrac{1}{2(\mu+1)} (an+b) \right).
  \]

  The stated result follows with the constant defined by
  \begin{equation}\label{eqn:gconst}
    \gconst = \tfrac{1}{2(\mu + 1)}.
  \end{equation}
\end{proof}

\paragraph{Outline of proof of Theorem~\ref{thm:oramstash}.}
We will mostly follow the proof of the small-stash-size theorem in Path
ORAM~\cite{CCS:SDSFRY13}.  The proof of the theorem consists of several steps.
\BEN
\item 
We recall the definition of $\infty$-ORAM (ORAM with infinitely large
buckets) and show that stash usage in
an $\infty$-ORAM with post-processing is the same as that in the actual
vORAM.

\item 
We rely on results from the most recent version of \cite{pathORAMarxiv}
to show that the stash usage after post-processing is greater than $R$ if and
only if there exists a subtree for which its usage in $\infty$-ORAM is more than
its capacity. 

\item
We bound the total size of all blocks in any such subtree by combining
two separate
measure concentrations on the number of blocks in any such subtree, and
the total size of any fixed number of variable-length blocks.

\item 
We complete the proof by connecting the measure concentrations to the
actual stash size, in a similar way to \cite{pathORAMarxiv}.

\EEN

Note that the first and third steps are those that differ most substantially from
prior work, and where we must incorporate the unique properties of
the vORAM.

\paragraph{Proof of Theorem~\ref{thm:oramstash}.} We now give the proof. 

  \paragraph{\it Step 1: $\infty$-ORAM.}
  The $\infty$-ORAM is the same as our vORAM tree, except that each
  bucket has infinite size. In any \writeback{} operation all blocks
  will go as far down along the path as their identifier allows.
  
  After simulating a series of vORAM operations on the $\infty$-ORAM, we
  perform a greedy post-processing to restore the block size
  condition:
  \begin{itemize}
  \item Repeatedly select a bucket storing more than $Z$ bytes. Remove a
  partial block from the bucket, and let $b$ be the number of remaining
  bytes stored in the bucket.

  \item If $Z-b$ is greater than the size of metadata per partial
  block (identifier and length), 
  then there is some room left in the bucket. In this case,
  split the removed block into two parts. Place the last $Z-b$ bytes
  into the current bucket and the remainder into the parent bucket. 
  Otherwise, if there is insufficient room in the bucket,
  place the entire block into the parent bucket, or into the
  stash if the current bucket is the root.
  \end{itemize}
  By continuing this process until there are no remaining buckets with
  greater than $Z$ bytes, we have returned to a normal vORAM with bucket
  size $Z$. Furthermore, there is an ordering of the accesses, with the
  same identifiers and block lengths, that would result in the same
  vORAM. Since the access order of the $\infty$-ORAM does not matter,
  this shows that the two models are equivalent after post-processing.

  Observe that we are ignoring the metadata (block identifiers and
  length strings). This is acceptable, as the removal process 
  in the actual vORAM always ensures that each partial block of a
  given block, except possibly for the first (highest in the vORAM
  tree), has size at least equal to the size of its metadata. In that
  way, at most half the vORAM is used for metadata storage, and so the
  metadata has only a constant factor effect on the overall performance.

  \paragraph{\it Step 2: Overflowing subtrees.}
  Consider the size of vORAM stash after any series of vORAM operations
  that result in a total of at most $n$ blocks being stored.
  Similarly to \cite[Lemma 2]{pathORAMarxiv}, the stash size at this point
  is equal to the total overflow from some subtree of the 
  $\infty$-ORAM buckets that contains the root. If we write $\tau$ for
  that subtree, then we have
  \[\begin{array}{l}
    |stash| > BR ~~~~~\text{iff}~~ \\
    ~~~  \sum_{\text{node } v\in\tau} (\mbox{size of $v$ in $\infty$-ORAM)}
    \ge Z|\tau| + BR.
    \end{array}
  \]

  \paragraph{\it Step 3: Size of subtrees.}
  We prove a bound on the total size of all blocks in any subtree $\tau$
  in the $\infty$-ORAM in two steps. First we bound the \emph{number} of
  blocks in the subtree, which can use the same analysis as the Path
  ORAM; then we bound the total \emph{size} of a given number of
  variable-length blocks; and, finally, we combine these with a union
  bound argument.

  To bound the total number of blocks that
  occur in $\tau$, because the block sizes do not matter in  the
  $\infty$-ORAM, we can simply recall from \cite[Lemma 5]{pathORAMarxiv}
  that, for any subtree $\tau$, the probability that $\tau$ contains 
  more than $5|\tau| + R/4$ blocks is at most
  \begin{equation}\label{eqn:sbcount}
    \tfrac{1}{4^{|\tau|}} \cdot (0.9332)^{|\tau|} \cdot (0.881)^R.
  \end{equation}

  Next we consider the total size of $5|\tau| + R/4$ variable-length
  blocks. From the statement of the theorem, each block size is
  stochastically dominated by $BX$, where $B$ is the expected block size
  and $X$ is a geometric random variable with expected value $\mu=1$.
  From Lemma~\ref{lem:sumgeom}, the total size 
  of all $5|\tau| + R/4$ blocks exceeds 
  $2(a(5|\tau|+R/4))B$ 
  with probability at most
  \[\exp\left(-\gconst a \left(5|\tau|+R/4\right)\right).\]

  From \eqref{eqn:gconst}, we can take $\gconst=1/4$, and by setting
  $a=2 > (4/5)\ln 4$, the probability that the total size of
  $5|\tau|+R/4$ blocks exceeds
  $(20|\tau| + R)B$
  is at most
  $\exp\left(-\tfrac{5}{2} |\tau| - \tfrac{1}{8} R\right)$,
  which in turn is less than
  \begin{equation}\label{eqn:sbsize}
    \tfrac{1}{4^{|\tau|}} \cdot (0.329)^{|\tau|} \cdot (0.883)^{R}.
  \end{equation}

  Finally, by the union bound, the probability that the total size of
  all blocks in $\tau$ exceeds $(20|\tau| + R)B$ is at most the sum of
  the probabilities in \eqref{eqn:sbcount} and \eqref{eqn:sbsize},
  which is less than
  \begin{equation}\label{eqn:sbtot}
    \tfrac{2}{4^{|\tau|}} \cdot (0.9332)^{|\tau|} \cdot (0.883)^R.
  \end{equation}

  \paragraph{\it Step 4: Stash overflow probability.}
  As in \cite[Section 5.2]{pathORAMarxiv}, the number of subtrees of
  size $i$ is less than $4^i$, and therefore by another application of
  the union bound along with \eqref{eqn:sbtot}, the probability of
  \emph{any} subtree $\tau$ having total block size greater than
  $(20|\tau|+R)B$ is at most
  \begin{align*}
    &\sum_{i\ge 1} 4^i \tfrac{2}{4^i} \cdot (0.9332)^i \cdot (0.883)^R \\
    &< 28 \cdot (0.883)^R.
  \end{align*}

\subsection{Proof of Theorem~\ref{thm:hirbperf}}

We now utilize Lemma~\ref{lem:sumgeom}
to prove the two lemmata on the distributions of the
number and size of HIRB tree nodes.

\begin{lemma}\label{lem:hirbcount}
  Suppose a HIRB tree with $n$ items has height $H\ge \log_\beta n$,
  and let $X$ be the total number of nodes in the HIRB, which is a
  random variable over the choice of hash function in initializing the
  HIRB. Then for any $m\ge 1$, we have
  \[\Pr\left[X \ge H + 4n + m \right] < 0.883^m.\]
\end{lemma}
In other words, the number of HIRB nodes in storage at any given time is
$O(n)$ with high probability.
The proof is a fairly standard application of the
Hoeffding inequality \cite{Hoe63}.  

\begin{proof}
  The HIRB has $H$ nodes initially. Consider the $n$ items
  $\dskey_1,\ldots,\dskey_n$ in the HIRB. Because the tree is
  uniquely represented, we can consider the number of nodes after 
  inserting the items in any particular order.

  When inserting an item with $\dskey_i$ into the HIRB, its height
  $h = \chooseheight(\dskey_i)$ is computed from the label hash,
  where $0\le h\le H$, and then exactly $h$ existing HIRB nodes are
  split when $\dskey_i$ is inserted, resulting in exactly $h$ newly
  created nodes.

  Therefore the total number of nodes in the HIRB after inserting all
  $n$ items is exactly $H$ plus the sum of
  the heights of all items in the HIRB, which from
  Assumption~\ref{ass:hirblevel} is the sum of $n$ iid geometric random
  variables, each with expected value $1/(\beta - 1)$. Call this sum
  $Y$.

  We are interested in bounding the probability that $Y$ exceeds
  $4n+m$. Writing
  $\mu=1/(\beta-1)$ for the expected value of each r.v., we have
  $\mu+1 = \beta/(\beta-1)$, which is at most $2$ since $\beta \ge 1$. 
  This means that
  $4n+m \ge (\mu+1)(2n + m/2)$, and from
  Lemma~\ref{lem:sumgeom},
  \begin{align*}
    \Pr[X\ge H + 4n + m] &= \Pr[Y\ge 4n+m] \\
      &\le \Pr[Y \ge (\mu+1)(2n+m/2)] \\
      &< \exp(-\gconst(2n + m/2)) \\
      &< \exp(-\gconst m/2).
  \end{align*}

  Because $\mu+1\le 2$,
  $\gconst = 1/(2(\mu+1)) \ge 1/4$. Numerical computation confirms that
  $\exp(-1/8) < 0.883$, which completes the proof.
\end{proof}

Along with the bound above on the number of HIRB nodes, we also need a
bound on the size of each node.

\begin{lemma}\label{lem:hirbnode}
  Suppose a HIRB tree with $n$ items has height $H \ge \log_\beta n$,
  and let $X$, a random variable over the choice of hash function,
  be the size of an arbitrary node in the HIRB.
  Then for any $m\ge 1$, we have
  \[\Pr[X \ge m\cdot\nodesize_\beta] < 0.5^m.\]
\end{lemma}

The proof of this lemma works by first bounding the probability that the number
of items in any node is at most $m\beta$ and applies the formula for node size,
i.e., 
\begin{equation}
\label{eqn:nodesize}
\begin{array}{l}
\nodesize_k  = \\
~~ (k+1)(2T+\hashparam+1) + k( |\Hash(\dskey)| +|\dsval|).
\end{array}
\end{equation} 

\begin{proof}
  We first show that the probability that any node's branching factor is
  more than $m\beta$ is at most $0.5^m$.
  This first part
  requires a special case for the root node, and a general case for any
  other node. Then we show that
  any node with branching factor at most 
  $m\beta$ has size less than $m\cdot\nodesize_\beta$.
  
  First consider the items in the root node. These items all
  have height $H$, which according to Assumption~\ref{ass:hirblevel}
  occurs for any given \dskey{} with probability
  $1/\beta^H$. Therefore the number of items in the root node follows a
  binomial distribution with parameter $1/\beta^H$. It is a standard
  result (for example, Theorem C.2 in \cite{CLRS01}) that a sample from
  such a distribution is at least $k$ with probability at most
  \[\binom{n}{k} \frac{1}{\beta^{Hk}}
    < \frac{n^k}{2^{k-1} \beta^{Hk}}.\]

  From the assumption $H\ge\log_\beta n$, $n^k \le \beta^{Hk}$, so the
  bound above becomes simply $2^{-k+1}$. Setting $k=m\beta$, the
  probability that the root node has at least $k$ items and hence
  branching factor greater than $m\beta$, is seen to be at most
  $2^{-m\beta+1}$, which is always at most $2^{-m}$ because $m\ge 1$ and
  $\beta \ge 2$..

  Next consider any nonempty HIRB tree node at height $\ell$, and
  consider a hypothetically infinite list of possible label hashes from
  the HIRB which have height at least $\ell$ and could be in this node.
  The actual number of items is determined by the number of those labels
  whose height is exactly equal to $\ell$ before we find one whose
  height is at least $\ell+1$. From Assumption~\ref{ass:hirblevel}, and
  the memorylessness property of the geometric distribution, these label
  heights are independent Bernoulli trials, and each height equals $\ell$ 
  with probability $(\beta-1)/\beta$.

  Therefore the size of each non-root node is a geometric random variable over
  $\{0,1,\ldots\}$ with parameter $1/\beta$. The probability that the
  node contains at least $m\beta$ items, and therefore has banching
  factor greater than $m\beta$, is exactly
  \[\left(\tfrac{\beta-1}{\beta}\right)^{m\beta} < \exp(-m) < 0.5^m.\]
  Here we use the fact that 
  $(1-\tfrac{1}{x})^{ax} < \exp(-a)$ for any
  $x\ge 1$ and any real $a$.

  All that remains is to say that a node with branching factor $m\beta$ has size
  less than $m\cdot\nodesize_\beta$, which follows directly from $m\ge 1$ and
  the definition of $\nodesize_\beta$ in (\ref{eqn:nodesize}). 
\end{proof}

Finally, we prove the main theorems on the vORAM+HIRB performance and
security.

\paragraph{Proof of Theorem~\ref{thm:hirbperf}.}
  We step through and motivate the choices of parameters, one by one.

  The expected branching factor $\beta$ must be at least 2 for the HIRB
  to work, which means we must always have $H\le \lg n$, and so
  $T = \lg(4n+\lg n+\hashparam) \le \lg(4n+H+\hashparam)$. Then
  Lemma~\ref{lem:hirbcount} guarantees that the number of HIRB nodes is
  less than $H+4n+\hashparam$ with probability at least
  $(0.883)^\hashparam$. This means that $T$ is an admissible height for
  the vORAM according to Theorem~\ref{thm:oramstash} with at least that
  probability.

  The choice of $\beta$ is such that 
  $Z \ge 20\cdot \nodesize_\beta$, using the inequality
  \[H \le \lg n < \lg(4n) < T.\]
  Therefore, by Lemma~\ref{lem:hirbnode}, the size of blocks in the HIRB
  will be admissible for the vORAM according to
  Theorem~\ref{thm:oramstash}.

  This allows us to say from the choice of $R$ and
  Theorem~\ref{thm:oramstash} that the probability of stash overflow is
  at most $28\cdot(0.883)^\hashparam$.

  Choosing $H$ as we do is required to actually apply
  Lemmas~\ref{lem:hirbcount} and \ref{lem:hirbnode} above.

  Finally, the probability of two label hashes in the HIRB colliding is
  at most $2^{-\hashparam}$. The stated result follows from the union
  bound over the three failure probabilities.

\end{appendices}

\end{document}